\documentclass[a4paper,UKenglish,cleveref,hyperref,autoref]{lipics-v2021}

\title{Stretch-width} %: Between Clique-Width and Twin-Width}

\titlerunning{Stretch-width} %: Between Clique-Width and Twin-Width}

\author{\'{E}douard Bonnet}{Univ Lyon, CNRS, ENS de Lyon, Université Claude Bernard Lyon 1, LIP UMR5668, France \and \url{http://perso.ens-lyon.fr/edouard.bonnet/}}{edouard.bonnet@ens-lyon.fr}{https://orcid.org/0000-0002-1653-5822}{}
\author{Julien Duron}{Univ Lyon, CNRS, ENS de Lyon, Université Claude Bernard Lyon 1, LIP UMR5668, France}{julien.duron@ens-lyon.fr}{}{}

\authorrunning{\'E. Bonnet, J. Duron}

\Copyright{Édouard Bonnet, Julien Duron}%TODO mandatory, please use full first names. LIPIcs license is "CC-BY";  http://creativecommons.org/licenses/by/3.0/

%TODO mandatory; please add comma-separated list of keywords

\category{}%optional, e.g. invited paper

\relatedversion{}%optional, e.g. full version hosted on arXiv, HAL, or other respository/website
\supplement{}%optional, e.g. related research data, source code,... hosted on a repository like zenodo, figshare, GitHub,...

\funding{This work was supported by the ANR projects TWIN-WIDTH (ANR-21-CE48-0014) and Digraphs (ANR-19-CE48-0013).}%optional, to capture a funding statement, which applies to all authors. Please enter author specific funding statements as fifth argument of the \author macro.
 
\acknowledgements{We thank Colin Geniet, Eunjung Kim, and Stéphan Thomassé for early discussions on the topic.}%optional

\nolinenumbers %uncomment to disable line numbering

\hideLIPIcs  %uncomment to remove references to LIPIcs series (logo, DOI,...), e.g. when preparing a pre-final version to be uploaded to arXiv or another public repository

\EventEditors{John Q. Open and Joan R. Access}
\EventNoEds{2}
\EventLongTitle{42nd Conference on Very Important Topics (CVIT 2016)}
\EventShortTitle{CVIT 2016}
\EventAcronym{CVIT}
\EventYear{2016}
\EventDate{December 24--27, 2016}
\EventLocation{Little Whinging, United Kingdom}
\EventLogo{}
\SeriesVolume{42}
\ArticleNo{23}

\usepackage[utf8]{inputenc}  

\usepackage[T1]{fontenc}
\usepackage{lmodern}

\usepackage[colorinlistoftodos,bordercolor=orange,backgroundcolor=orange!20,linecolor=orange,textsize=normalsize]{todonotes}

\usepackage{amsmath}  
\usepackage{amssymb}
\usepackage{bbm}
\usepackage{accents}
\usepackage{complexity}
\usepackage{booktabs}
\usepackage{fixmath}
\usepackage{paralist}

\makeatletter
\newtheorem*{rep@theorem}{\rep@title}
\newcommand{\newreptheorem}[2]{%
\newenvironment{rep#1}[1]{%
 \def\rep@title{#2 \ref{##1}}%
 \begin{rep@theorem}}%
 {\end{rep@theorem}}}
\makeatother

\newreptheorem{theorem}{Theorem}
\newreptheorem{lemma}{Lemma}
\newreptheorem{corollary}{Corollary}
\newreptheorem{proposition}{Proposition}
\newreptheorem{definition}{Definition}
\newreptheorem{remark}{Remark}

\crefname{observation}{Observation}{Observations}

\usepackage{xspace}
\usepackage{tikz}

\usepackage[ruled,vlined,linesnumbered]{algorithm2e}

\usetikzlibrary{fit}
\usetikzlibrary{arrows}
\usetikzlibrary{patterns}
\usetikzlibrary{calc}
\usetikzlibrary{shapes}
\usetikzlibrary{positioning}
\usetikzlibrary{math}
\usetikzlibrary{shapes.symbols}
\usetikzlibrary{decorations.pathreplacing}
\tikzset{draw half paths/.style 2 args={%
  decoration={show path construction,
    lineto code={
      \draw [#1] (\tikzinputsegmentfirst) -- 
         ($(\tikzinputsegmentfirst)!0.5!(\tikzinputsegmentlast)$);
      \draw [#2] ($(\tikzinputsegmentfirst)!0.5!(\tikzinputsegmentlast)$)
        -- (\tikzinputsegmentlast);
    }
  }, decorate
}}

\usepackage[scr=boondox,scrscaled=1.05]{mathalfa}

\renewcommand{\geq}{\geqslant}
\renewcommand{\leq}{\leqslant}
\renewcommand{\preceq}{\preccurlyeq}

\newcommand{\mis}{\textsc{Max Independent Set}\xspace}

\newcommand{\smis}{\textsc{MIS}\xspace}

\newcommand{\tww}{\text{tww}\xspace}
\newcommand{\tw}{\text{tw}\xspace}
\newcommand{\sn}{\text{sn}\xspace}
\newcommand{\ctww}{\text{ctww}\xspace}
\newcommand{\bw}{\text{bandw}\xspace}
\newcommand{\cutw}{\text{cutw}\xspace}

\newcommand{\cw}{\text{cw}\xspace}
\newcommand{\stw}{\text{stw}\xspace}

\newcommand{\rows}{\text{rows}}
\newcommand{\cols}{\text{columns}}

\newcommand{\at}[2]{A_{#1}(#2)}
\newcommand{\Root}[2]{N_{#1}(#2)}
\newcommand{\tree}[2]{T_{#1}(#2)}
\newcommand{\atG}[2]{G(A_{#1}(#2))}

\newcommand{\conv}[1]{\text{conv}(#1)}
\newcommand{\cross}{conflict with\xspace}
\newcommand{\crosses}{conflicts with\xspace}
\newcommand{\crossing}{conflicting with\xspace}

\newcommand{\confing}{conflicting\xspace}

\newcommand{\interf}{interfere\xspace}
\newcommand{\interfs}{interferes\xspace}
\newcommand{\interfing}{interfering\xspace}

\newcommand{\stret}{stretch\xspace}
\newcommand{\str}[1]{\text{str}(#1)}

\newcommand{\itsubd}{iterated subdivision\xspace}

\newcommand{\CC}{\mathcal C\xspace}
\renewcommand{\P}{\mathcal P\xspace}
\renewcommand{\R}{\mathcal R\xspace}

\newcommand{\Ov}{\text{Ov}}

\theoremstyle{definition}

\begin{document}

\maketitle

\begin{abstract}
  We introduce a new parameter, called stretch-width, that we show sits strictly between clique-width and twin-width.
  Unlike the reduced parameters [BKW '22], planar graphs and polynomial subdivisions do not have bounded stretch-width.
  This leaves open the possibility of efficient algorithms for a broad fragment of problems within Monadic Second-Order (MSO) logic on graphs of bounded stretch-width. 
  In this direction, we prove that graphs of bounded maximum degree and bounded stretch-width have at most logarithmic treewidth.
  As a consequence, in classes of bounded stretch-width, \textsc{Maximum Independent Set} can be solved in subexponential time $2^{\Tilde{O}(n^{4/5})}$ on $n$-vertex graphs, and, if further the maximum degree is bounded, Existential Counting Modal Logic [Pilipczuk '11] can be model-checked in polynomial time. 
  We also give a~polynomial-time $O(\text{OPT}^2)$-approximation for the stretch-width of symmetric $0,1$-matrices or ordered graphs.

  Somewhat unexpectedly, we prove that exponential subdivisions of bounded-degree graphs have bounded stretch-width.
  This allows to complement the logarithmic upper bound of treewidth with a~matching lower bound. 
  We leave as open the existence of an efficient approximation algorithm for the stretch-width of unordered graphs, if the exponential subdivisions of all graphs have bounded stretch-width, and if graphs of bounded stretch-width have logarithmic clique-width (or rank-width).
\end{abstract}

\section{Introduction}\label{sec:intro}

Various graph classes have bounded \emph{twin-width}\footnote{We refer the reader to~\cref{sec:prelim} for the relevant definitions.} such as, for instance, bounded clique-width graphs, proper minor-closed classes, proper hereditary subclasses of permutation graphs, and some expander classes~\cite{twin-width1}.
Low twin-width, together with the witnessing \emph{contraction sequences}, enables parameterized algorithms (that are unlikely in general graphs) for testing if a~graph satisfies a first-order sentence~\cite{twin-width1,twin-width3}, and improved approximation algorithms for highly inapproximable packing and coloring problems~\cite{twin-width3,Berge23}.

However one should not expect a large gain, in the low twin-width regime, as far as (non-parameterized) exact algorithms are concerned.
This is because every graph obtained by subdividing (at least) $2 \lceil \log n \rceil$ times each edge of an $n$-vertex graph $G$ has twin-width at most~4~\cite{Berge21}.
It was already observed in the 70's that a problem like \textsc{Maximum Independent Set} (\smis, for short) remains NP-complete in $2t$-subdivisions~\cite{Poljak74}.
Furthermore, known reductions~\cite{GJ79} combined with the Sparsification Lemma~\cite{sparsification}, imply that unless the Exponential-Time Hypothesis\footnote{the assumption that there is a $\lambda > 1$ such that $n$-variable \textsc{3-SAT} cannot be solved in time $\lambda^n n^{O(1)}$} (ETH) fails~\cite{Impagliazzo01} solving \smis in subcubic graphs requires time $2^{\Omega(n)}$.
The previous remarks entail that, unless the ETH fails, solving \smis in subcubic graphs of twin-width at most~4 requires time $2^{\Omega(n/\log n)}$.

In contrast, on the significantly less general classes of bounded clique-width not only can \smis be solved in polynomial-time, but a~fixed-parameter algorithm solving MSO$_1$\footnote{Monadic Second-Order logic, when the second-order variables can only be vertex subsets} model checking in time $f(w,|\varphi|)n^{O(1)}$ exists~\cite{Courcelle00,Oum08}, with $w$ the clique-width of the input graph, $\varphi$ the input sentence, and $f$ some computable function.

In this paper, we start exploring the trade-off between class broadness and algorithmic generality in the zone delimited by bounded clique-width and bounded twin-width.   
It~may seem like the \emph{reduced parameters}~\cite{reduced-bdw}, where a graph has \emph{reduced $p$} at most $k$ if it admits a~contraction sequence in which all the red graphs have parameter $p$ at most $k$, are exactly designed to tackle this endeavor.
Indeed by definition, twin-width is \emph{reduced $\Delta$}, where $\Delta$ is the maximum degree, and it was shown that \emph{reduced maximum connected component size} (under the name of \emph{component twin-width}) is functionally equivalent to clique-width~\cite{twin-width6}.
Between \emph{maximum connected component size} and \emph{maximum degree}, there are several parameters~$p$, such as bandwidth, cutwidth, treewidth$+\Delta$, whose reduced parameters give rise to a~strict~\cite{reduced-bdw} hierarchy between bounded clique-width and bounded twin-width.
Unfortunately, even reduced bandwidth --the closest to clique-width among the above-mentioned reduced parameters-- turns out to be too general in the following sense: the $n$-subdivision of any $n$-vertex graph has reduced bandwidth at most~2~\cite{reduced-bdw}.
This means, by the arguments of the second paragraph of this introduction, that solving \smis on graphs of bounded reduced bandwidth requires time $2^{\Omega(\sqrt n)}$, unless the ETH fails, even among graphs of bounded degree.
Actually, another fact leading to the same conclusion is that planar graphs have bounded reduced bandwidth~\cite{reduced-bdw}.

We therefore introduce another parameter, that we call stretch-width\footnote{We refer a reader who would already want a formal definition to the start of~\cref{sec:stw}.} and denote by $\stw$, which, while inspired by reduced parameters, does not fully fit that framework.
To a~first approximation, \emph{stretch-width} can be thought as \emph{reduced bandwidth} where the bandwidth upper bound on the red graphs have to be witnessed by a~single (and fixed) order on the vertex set.
Observe indeed that the linear orders witnessing that all the red graphs of the sequence have low bandwidth can, in reduced bandwidth, be very different one from the other.
We first show that the family of bounded stretch-width classes strictly contains the family of bounded clique-width classes.
Using an upper bound of component twin-width by clique-width~\cite{BarilSlides}, we prove that:

\begin{theorem}\label{thm:stw-gen-cw}
  The stretch-width of any graph is at most twice its clique-width.
\end{theorem}

Then we provide a~separating class of bounded stretch-width and unbounded clique-width.

\begin{theorem}\label{thm:separation}
  There is an infinite family of graphs $G$ with bounded stretch-width and clique-width in $\Omega(\log |V(G)|/ \log \log |V(G)| )$.
\end{theorem}

After establishing~\cref{thm:subd}, we even get a simpler construction with bounded degree and stretch-width, but treewidth and clique-width in $\Omega(\log |V(G)|)$.
As was done for twin-width~\cite{twin-width4}, we give an effective characterization of bounded stretch-width for symmetric $0,1$-matrices (or ordered graphs).
\begin{theorem}\label{thm:matrix-characterization}
  A class $\mathcal C$ of symmetric $0,1$-matrices has bounded stretch-width if and only if there is an integer $k$ such that no matrix of $\mathcal C$ has a $k$-wide division.
\end{theorem}
The~\emph{$k$-wide division} (see definition in~\cref{sec:matrices}) is a scaled-down version of the \emph{$k$-rich division} that analogously characterizes matrices of bounded twin-width~\cite{twin-width4}.
\cref{thm:matrix-characterization} yields a~polynomial-time approximation algorithm for the stretch-width of symmetric $0,1$-matrices.
More precisely:
\begin{theorem}\label{thm:approx-stw}
  Given an integer $k$ and a symmetric $n \times n$ $0,1$-matrix $M$, there is an $n^{O(1)}$-time algorithm that outputs a~sequence witnessing that $\stw(M) = O(k^3)$ or correctly reports that $\stw(M) > k$.
\end{theorem}
Compared to the approximation algorithm for the twin-width of a~matrix, this is better both in terms of running time (polynomial vs fixed-parameter tractable) and approximation factor (quadratic vs exponential).   

\medskip

Conveniently for the sought algorithmic applications, planar graphs and $n^c$-subdivisions of $n$-vertex graphs (for any constant $c$) both have \emph{unbounded} stretch-width (whereas they have bounded reduced bandwidth if $c \geqslant 1$).
We indeed establish the following upper bound on treewidth, implying that graphs of bounded maximum degree and bounded stretch-width have at most logarithmic treewidth. 

\begin{theorem}\label{thm:tw-bound}
There is a~$c$ such that for every graph $G$, $\tw(G) \leqslant c \Delta(G)^4 \stw(G)^2 \log |V(G)|$.
\end{theorem}

We match~\cref{thm:tw-bound} with a~lower bound.
There are graphs with bounded $\Delta+\stw$ and treewidth growing as a~logarithm of their number of vertices.
This is because, as we prove, very long subdivisions of bounded-degree graphs have bounded stretch-width.

\begin{theorem}\label{thm:subd}
  Every $(\geqslant n2^m)$-subdivision of every $n$-vertex $m$-edge graph $G$ of maximum degree $d$ has stretch-width at most $32(4d+5)^3$.  
\end{theorem}

By \emph{$(\geqslant s)$-subdivision} of $G$, we mean every graph obtained by subdividing each edge of $G$ at least $s$~times.
In particular, for every natural $k$, the $n$-vertex $k^2 2^{2k(k-1)}$-subdivision of the $k \times k$ grid has bounded maximum degree (by 4) and stretch-width (by 296352), whereas it has treewidth $k=\Omega(\sqrt{\log n})$.
A~more careful argument and reexamination of~\cref{thm:subd} show that, for some constant $c$, the $n$-vertex $2^{ck}$-subdivision of the $k \times k$ grid has bounded $\Delta+\stw$, and treewidth $k=\Omega(\log n)$ matching the upper bound of~\cref{thm:tw-bound}.

The proofs of~\cref{thm:tw-bound,thm:subd} involve the notion of \emph{overlap graph} of a~graph $G$ whose vertex set is totally ordered by $\prec$, with one vertex per edge of $G$, and an edge between two ``overlapping edges'' of $G$, that is, two edges $ab$ and $cd$ such that $a \prec c \prec b \prec d$.
Using~\cref{thm:matrix-characterization}, we show that finding a~vertex ordering such that the overlap graph has no large biclique allows to bound the stretch-width. 
\begin{lemma}\label{lem:weakly-sparse-overlap}
  For every ordered graph $(G,\prec)$ and every integer $t$, if the overlap graph of $(G,\prec)$ has no $K_{t,t}$ subgraph then $\stw(G) < 32(2t+1)^3$.
\end{lemma}
\Cref{thm:subd} is then derived by designing a~long subdivision process that, for ordered graphs of maximum degree~$d$, reduces the bicliques in the overlap graph to a~size at~most linear in~$d$.   

\Cref{thm:tw-bound} has direct algorithmic implications for classes of bounded stretch-width.
\begin{proposition}\label{prop:subexp}
  There is an algorithm that solves \mis in graphs of bounded stretch-width with running time $2^{\Tilde{O}(n^{4/5})}$. 
\end{proposition}

Pilipczuk~\cite{Pilipczuk11} showed that any problem expressible in Existential Counting Modal Logic (ECML) admits a single-exponential fixed-parameter algorithm in treewidth.
In particular, ECML model checking can be solved in polynomial time in any class with logarithmic treewidth. 
This logic allows existential quantifications over vertex and edge sets followed by an \emph{arithmetic formula} and a \emph{counting modal formula} that shall be satisfied from every vertex~$v$.
The arithmetic formula is a~quantifier-free expression that may involve the cardinality of the vertex and edge sets, as well as integer parameters.
Counting modal formulas enrich quantifier-free Boolean formulas with $\Diamond^S \varphi$, whose semantics is that the current vertex~$v$ has a number of neighbors satisfying $\varphi$ in a~prescribed ultimately periodic set $S$ of non-negative integers.

The logic ECML+C gives to ECML the power of also using in the arithmetic formula the number of connected components in subgraphs induced by some vertex or edge sets.
There is a~Monte-Carlo polynomial-time algorithm for ECML+C in graphs of treewidth at most a~logarithm function in their number of vertices~\cite{Pilipczuk11}.
Most NP-hard graphs problems, such as \textsc{Maximum Independent Set}, \textsc{Minimum Dominating Set}, \textsc{Steiner Tree}, etc. are expressible in ECML+C; see \cite[Appendix D]{Pilipczuk11} for the ECML+C formulation of several examples. 

\begin{corollary}\label{cor:ecml}
  Problems definable in ECML (resp. ECML+C) can be solved in polynomial time (resp. randomized polynomial time) in bounded-degree graphs of bounded stretch-width.
\end{corollary}

\textbf{Perspectives.}
\cref{prop:subexp,cor:ecml} constitute some preliminary pieces of evidence of the algorithmic amenability of classes of bounded stretch-width.
We ask several questions.
How can the running time of~\cref{prop:subexp} be improved?
(As far as we know, there could be a polynomial-time algorithm for any problem defined in ECML on graphs of bounded stretch-width.)
As for twin-width, an approximation algorithm for stretch-width of (unordered) graphs remains open.
\Cref{lem:weakly-sparse-overlap} gives some hope that this question might be easier than its twin-width counterpart, especially among \emph{sparse} graphs.

Can we lift the bounded-degree requirement in \cref{thm:subd}, that is, is there a~function $f$ and a~constant $c$, such that the stretch-width of any~$(\geqslant f(n))$-subdivision of any $n$-vertex graph is at~most~$c$?
Our separating example showing that bounded stretch-width is strictly more general than bounded clique-width (\cref{thm:separation}) yields graphs of essentially logarithmic clique-width.
Is that true in general?
\begin{conjecture}\label{conj:log-cw}
  For every class $\mathcal C$ of bounded stretch-width, there is a constant $c$ such that for every $n$-vertex graph $G \in \mathcal C$ the clique-width of $G$ is at most $c \log n$. 
\end{conjecture}
We ask the same question with \emph{rank-width} instead of \emph{clique-width}, which would be more algorithmically helpful.
One interpretation of \cref{thm:tw-bound} is that graphs of bounded maximum degree and bounded stretch-width have logarithmic treewidth.
Whether the \emph{bounded-degree} constraint can be relaxed to the mere absence of large bipartite complete subgraphs is related to~\cref{conj:log-cw}.
A~positive answer to~\cref{conj:log-cw} would indeed imply this relaxation, as Gurski and Wanke have shown that graphs without $K_{t,t}$ subgraphs have treewidth at most their clique-width times~$3t$~\cite{Gurski00}.
A~natural future work would consist of using the witness of low stretch-width to get improved algorithms compared to those attained with a witness of low twin-width. 

\medskip

\textbf{Related work.}
Our work is in line with twin-width~\cite{twin-width1}, and the reduced parameters~\cite{reduced-bdw}.
\Cref{thm:stw-gen-cw} closely follows a~similar proof in the sixth paper of the twin-width series~\cite{twin-width6}, while \cref{thm:matrix-characterization} is inspired by the fourth paper~\cite{twin-width4}, and notably the so-called \emph{rich divisions}.

Finding the right logic for a~given width parameter, or the right width parameter for a~given logic has been a~common goal ever since Courcelle's and Courcelle-Makowsky-Rotics's theorems~\cite{Courcelle90,Courcelle00} relating treewidth with MSO$_2$, and clique-width with MSO$_1$.
Recent developments (all from 2023) include an efficient model checking of the new logic \textsc{A\&C~DN} (an extension of Existential MSO$_1$) on classes of bounded mim-width~\cite{Bergougnoux23}, the new parameter flip-width~\cite{Torunczyk23}, which could lead to an efficient first-order (FO) model checking in a~very general class, and efficient model checking algorithms for FO extensions with disjoint-paths predicates in proper minor-closed classes~\cite{Golovach23}, and in proper topological-minor-closed classes~\cite{Schirrmacher23}.

Classes with logarithmic treewidth, although not a priori defined as such, are somewhat rare.
To our knowledge, the first such example is the class of triangle-free graphs with no theta (see~\cite{Sintiari21} for the lower bound, and \cite{Abrishami22}, for the upper bound).
Another example consists of graphs without $K_{t,t}$ subgraph and bounded induced cycle packing number~\cite{Bonamy23}.
We add a~new family: graphs of bounded maximum degree and bounded stretch-width.
Note that these three families are pairwise incomparable.

\section{Preliminaries}\label{sec:prelim}

For $i \leqslant j$ two integers, we denote the set of integers that are at least $i$ and at most $j$ by $[i,j]$, and $[i]$ is a~short-hand for $[1,i]$.
We use the standard graph-theoretic notations.
In particular, for a graph $G$, we denote by $V(G)$ its set of vertices and by $E(G)$ its set of edges.
If $S \subseteq V(G)$, the \emph{subgraph of $G$ induced by $S$}, denoted $G[S]$ is the graph obtained from $G$ by removing the vertices not in $S$.

\subsection{Contraction sequences and twin-width}

Twin-width is a graph parameter introduced by Bonnet, Kim, Thomassé, and Watrigant~\cite{twin-width1}.
A~possible definition involves the notions of \emph{trigraphs}, \emph{red graphs}, and \emph{contraction sequences}.
A~\emph{trigraph} is a graph with two types of edges: black (regular) edges and red (error) edges.
The \emph{red graph} $\R(H)$ of a trigraph $H$ consists of ignoring its black edges, and considering its red edges as being normal (black) edges.
We may say~\emph{red neighbor} (or \emph{red neighborhood}) to simply mean a~neighbor (or neighborhood) in the red graph.
A~(vertex) \emph{contraction} consists of merging two (non-necessarily adjacent) vertices, say, $u, v$ into a~vertex~$w$, and keeping every edge $wz$ black if and only if $uz$ and $vz$ were previously black edges.
The other edges incident to $w$ become red (if not already), and the rest of the trigraph remains the same.
A~\emph{contraction sequence} of an $n$-vertex graph $G$ is a sequence of trigraphs $G=G_n$, $\ldots, G_1$ such that $G_i$ is obtained from $G_{i+1}$ by performing one contraction.
A~\mbox{\emph{$d$-sequence}} is a contraction sequence in which every vertex of every trigraph has at most $d$ red edges incident to it.
In other words, every red graph of the sequence has maximum degree at most~$d$.
The~\emph{twin-width} of~$G$, denoted by $\tww(G)$, is then the minimum integer~$d$ such that $G$ admits a $d$-sequence.
\Cref{fig:contraction-sequence} gives an example of a graph with a 2-sequence, i.e., of twin-width at most~2.

\begin{figure}[h!]
  \centering
  \resizebox{400pt}{!}{
  \begin{tikzpicture}[
      vertex/.style={circle, draw, minimum size=0.68cm}
    ]
    \def\s{1.2}
    %G=G7
    \foreach \i/\j/\l in {0/0/a,0/1/b,0/2/c,1/0/d,1/1/e,1/2/f,2/1/g}{
      \node[vertex] (\l) at (\i * \s,\j * \s) {$\l$};
    }
    \foreach \i/\j in {a/b,a/d,a/f,b/c,b/d,b/e,b/f,c/e,c/f,d/e,d/g,e/g,f/g}{
      \draw (\i) -- (\j);
    }

    %G6=G7/e,f
    \begin{scope}[xshift=3 * \s cm]
    \foreach \i/\j/\l in {0/0/a,0/1/b,0/2/c,1/0/d,2/1/g}{
      \node[vertex] (\l) at (\i * \s,\j * \s) {$\l$};
    }
    \foreach \i/\j/\l in {1/1/e,1/2/f}{
      \node[vertex,opacity=0.2] (\l) at (\i * \s,\j * \s) {$\l$};
    }
    \node[draw,rounded corners,inner sep=0.01cm,fit=(e) (f)] (ef) {ef};
    \foreach \i/\j in {a/b,a/d,b/c,b/d,b/ef,c/ef,c/ef,d/g,ef/g,ef/g}{
      \draw (\i) -- (\j);
    }
    \foreach \i/\j in {a/ef,d/ef}{
      \draw[red, very thick] (\i) -- (\j);
    }
    \end{scope}

    %G5=G6/a,d
    \begin{scope}[xshift=6 * \s cm]
    \foreach \i/\j/\l in {0/1/b,0/2/c,2/1/g,1/1/ef}{
      \node[vertex] (\l) at (\i * \s,\j * \s) {$\l$};
    }
    \foreach \i/\j/\l in {0/0/a,1/0/d}{
      \node[vertex,opacity=0.2] (\l) at (\i * \s,\j * \s) {$\l$};
    }
    \draw[opacity=0.2] (a) -- (d);
    \node[draw,rounded corners,inner sep=0.01cm,fit=(a) (d)] (ad) {ad};
    \foreach \i/\j in {ad/b,b/c,b/ad,b/ef,c/ef,c/ef,ef/g,ef/g}{
      \draw (\i) -- (\j);
    }
    \foreach \i/\j in {ad/ef,ad/g}{
      \draw[red, very thick] (\i) -- (\j);
    }
    \end{scope}

     %G4=G5/b,ef
    \begin{scope}[xshift=9 * \s cm]
    \foreach \i/\j/\l in {0/2/c,2/1/g,0.5/0/ad}{
      \node[vertex] (\l) at (\i * \s,\j * \s) {$\l$};
    }
    \foreach \i/\j/\l in {0/1/b,1/1/ef}{
      \node[vertex,opacity=0.2] (\l) at (\i * \s,\j * \s) {$\l$};
    }
    \draw[opacity=0.2] (b) -- (ef);
    \node[draw,rounded corners,inner sep=0.01cm,fit=(b) (ef)] (bef) {bef};
    \foreach \i/\j in {ad/bef,bef/c,bef/ad,c/bef,c/bef,bef/g}{
      \draw (\i) -- (\j);
    }
    \foreach \i/\j in {ad/bef,ad/g,bef/g}{
      \draw[red, very thick] (\i) -- (\j);
    }
    \end{scope}

     %G3=G4/a,dg
    \begin{scope}[xshift=11.7 * \s cm]
    \foreach \i/\j/\l in {0/2/c}{
      \node[vertex] (\l) at (\i * \s,\j * \s) {$\l$};
    }
     \foreach \i/\j/\l in {0.5/0/adg,0.5/1.1/bef}{
      \node[vertex] (\l) at (\i * \s,\j * \s) {\footnotesize{\l}};
    }
    \foreach \i/\j in {c/bef}{
      \draw (\i) -- (\j);
    }
    \foreach \i/\j in {adg/bef}{
      \draw[red, very thick] (\i) -- (\j);
    }
    \end{scope}

    %G2=G3/c,bef
    \begin{scope}[xshift=13.7 * \s cm]
    \foreach \i/\j/\l in {0.5/0/adg,0.5/1.1/bcef}{
      \node[vertex] (\l) at (\i * \s,\j * \s) {\footnotesize{\l}};
    }
    \foreach \i/\j in {adg/bcef}{
      \draw[red, very thick] (\i) -- (\j);
    }
    \end{scope}

    %G1=G2/adg,bcef
    \begin{scope}[xshift=15 * \s cm]
    \foreach \i/\j/\l in {1/0.75/abcdefg}{
      \node[vertex] (\l) at (\i * \s,\j * \s) {\tiny{\l}};
    }
    \end{scope}
    
  \end{tikzpicture}
  }
  \caption{A 2-sequence witnessing that the initial graph has twin-width at most~2.}
  \label{fig:contraction-sequence}
\end{figure}

\subsection{Partition sequences}

Partition sequences yield an equivalent viewpoint to contraction sequences.
Instead of dealing with a sequence of trigraphs $G=G_n, \ldots, G_1$, we now have a~sequence of partitions $\P_n, \ldots, \P_1$ of $V(G)$, with $\P_n = \{\{v\}~|~v\in V(G)\}$ and for every $i \in [n-1]$, $\P_i$ is obtained from $\P_{i+1}$ by merging two parts $X, Y \in \P_{i+1}$ into one ($X \cup Y$).
In particular $\P_1=\{V(G)\}$.
Now one can obtain the red graph $\R(G_i)$ of $G_i$, as the graph whose vertices are the parts of $\P_i$, and whose edges link two parts $X \neq Y \in P_i$ whenever there is $u, u' \in X$ and $v, v' \in Y$ such that $uv \in E(G)$ and $u'v' \notin E(G)$.
We may call two such parts $X, Y$ \emph{inhomogeneous}.
On the contrary, two parts $X, Y$ are \emph{homogeneous} in $G$ when every vertex of $X$ is adjacent to every vertex of $Y$, or no vertex of $X$ is adjacent to a~vertex of $Y$.
We will also denote~$\R(G_i)$ by~$\R(\P_i)$.

\subsection{Reduced parameters and functional equivalence}

We detail the definition of reduced parameters mainly for the introduction, and the notions of \emph{functional equivalence}, \emph{component of twin-width}, and \emph{reduced bandwidth}; the latter being conceptually close to stretch-width.
The reduced parameters will not be useful, \emph{per se}, in the rest of the paper.  

As we mentioned in the introduction, there are stronger\footnote{Note that every graph admits a sequence where all the red graphs consist of one star together with isolated vertices (namely, any partition sequence having, at every step, only one part that is not a~singleton). Stars form arguably the simplest class of unbounded degree. Thus trading the condition of maximum degree to an incomparable property on the red graphs should likely be accompanied by some extra requirement, like forcing the partitions to be reasonably ``balanced.''} constraints that one can put on the red graphs than merely having bounded maximum degree.
This leads to the \emph{reduced parameters} as defined in~\cite{reduced-bdw}.
If $p$ is a~graph parameter, one can define the parameter \emph{reduced~$p$}, denoted by $p^{\downarrow}$, of a~$n$-vertex graph $G$ is the minimum over every contraction sequence $G=G_n, \ldots, G_1$ of $\max_{i \in [n]}p(\R(G_i))$.
Then $\Delta^\downarrow$ is the twin-width, when $\Delta$ denotes the maximum degree.
For $p=\star$, the maximum size of a connected component, $p^\downarrow$ is the so-called \emph{component twin-width} (see~\cite{twin-width6}).

A~graph parameter $p$ is \emph{functionally bounded} by a~graph parameter $q$, denoted by $p \sqsubseteq q$, if there is a~function $f$ such that for every graph $G$, $p(G) \leqslant f(q(G))$.
Parameters $p, q$ are \emph{functionally equivalent} or \emph{tied} if $p \sqsubseteq q$ and $q \sqsubseteq p$.
We finally denote by $p \sqsubset q$ the fact that $p \sqsubseteq q$ holds but $q \sqsubseteq p$ does not.
It is shown in~\cite{reduced-bdw} that, under some relatively mild assumptions, the strict ``inclusion'' $p \sqsubset q$ implies the strict ``inclusion'' $p^\downarrow \sqsubset q^\downarrow$.

It can be seen that $\star \sqsubset \bw \sqsubset \cutw \sqsubset (\Delta+\tw) \sqsubset \Delta$, where $\bw$, $\cutw$, $\tw$ are the bandwidth, cutwidth, and treewidth, respectively.
We give a~definition of bandwidth here (mainly because of the apparent similarity between stretch-width and reduced bandwidth).
The treewidth of a graph is defined in the next subsection, while we omit the definition of \emph{cutwidth} as we will not need it.
A~\emph{linear layout} of an $n$-vertex graph $G$ is a~bijective map $\sigma$ from $V(G)$ to $[n]$.
The~\emph{length} of an edge $e=uv \in E(G)$ under the linear layout $\sigma$ is defined as $|\sigma(u)-\sigma(v)|$.
The \emph{bandwidth} $\bw(G)$ of a graph $G$ is the minimum over every linear layout $\sigma$ of the maximum length of an edge $e \in E(G)$ under $\sigma$.

By~\cite{reduced-bdw}, $\ctww=\star^\downarrow \sqsubset \bw^\downarrow \sqsubset \cutw^\downarrow \sqsubset (\Delta+\tw)^\downarrow \sqsubset \Delta^\downarrow=\tww$.
As component twin-width $\ctww$ and clique-width $\cw$ are functionally equivalent~\cite{twin-width6}, we get a~strict ladder of classes interpolating between clique-width and twin-width.
However, for every parameter~$p$ considered so far except $\star$ (even bandwidth), the classes $\{G^{(n)}~|~n~\in~\mathbb N,~G~\text{has}~n~\text{vertices}\}$ and of all planar graphs have bounded $p^\downarrow$~\cite{reduced-bdw}, where $G^{(n)}$ denotes the $n$-subdivision of~$G$, that is, the graph obtained after replacing every edge of~$G$ by a~$(n+1)$-edge path.

As this is an obstacle to exactly solving more general problems than first-order model checking, the current paper is about a~new parameter $\stw$ (stretch-width) satisfying, as we will prove, $\cw \sqsubset \stw \sqsubset \bw^\downarrow$, while not containing the class of all $n$-subdivisions nor the one of all planar graphs. 

\subsection{Treewidth, separation number, and clique-width}

We recall the definition of treewidth and clique-width, for completeness.
We also state a~useful characterization of bounded treewidth in terms of balanced separators. 

A~\emph{tree-decomposition} of a~graph $G$ is a~pair $(T,\beta)$ where $T$ is a~tree, and $\beta$ is a~map from $V(T)$ to $2^{V(G)}$, such that the following three conditions are met:
\begin{compactitem}
\item for every $v \in V(G)$, there is a~$t \in V(T)$ such that $v \in \beta(t)$;
\item for every $uv \in E(G)$, there is a~$t \in V(T)$ such that $\{u,v\} \subseteq \beta(t)$;
\item for every $v \in V(G)$, $\{t \in V(T)~|~v \in \beta(t)\}$ induces a connected graph in $T$ (i.e., a~subtree).
\end{compactitem}

When dealing with treewidth in~\cref{sec:tw-bound} it will more convenient to think of it in terms of the functionally equivalent \emph{separation number}.
A~\emph{separation} $(A,B)$ of a~graph $G$ is such that $A \cup B=V(G)$ and there is no edge between $A \setminus B$ and $B \setminus A$.
The \emph{order} of the separation $(A,B)$ is $|A \cap B|$.
A~separation $(A,B)$ is \emph{balanced} if $\max(|A \setminus B|, |B \setminus A|) \leqslant \frac{2}{3} |V(G)|$.
The \emph{separation number} $\sn(G)$ of $G$ is the smallest integer $s$ such that every subgraph of $G$ admits a~balanced separation of order at most~$s$.
It is not difficult to show that for every graph $G$, $\sn(G) \leqslant \tw(G)+1$.
Dvorák and Norin showed the converse linear dependence:
\begin{lemma}[\cite{Dvorak19}]\label{lem:tw-sn}
  For every graph $G$, $\tw(G) \leqslant 15 \sn(G)$.
\end{lemma}

Note that if for some positive constant $c<1$, every subgraph $H$ of $G$ has a~separation $(A,B)$ that is \emph{$c$-balanced}, in the sense that $\max(|A \setminus B|, |B \setminus A|) \leqslant c |V(H)|$ of order at most~$s$, then every subgraph of $G$ has a~balanced separation of order $\lceil \frac{\log c}{\log(2/3)} \rceil \cdot s$.
In particular, by~\cref{lem:tw-sn}, $\tw(G) = O(s)$.

We finish~\cref{sec:prelim} with a~brief definition of clique-width, just for completeness, because the introduction contains several occurrences of it. 
This definition can be ignored in the rest of the paper, and clique-width thought as the reduced parameter \emph{component twin-width} (or $\star^\downarrow$).
The \emph{clique-width} $\cw(G)$ of a graph $G$ is the least number $k$ of colors, called \emph{labels}, needed to build $G$ from the following operations:
\begin{compactitem}
\item create a~vertex with a label $i \in [k]$,
\item make the union of two labeled graphs,
\item relabel every vertex colored $i$ with the label $j$, for some $i \neq j \in [k]$,
\item add all edges between vertices labeled $i$ and vertices labeled $j$, for some $i \neq j \in [k]$.
\end{compactitem}

\subsection{Outline}

In~\cref{sec:stw}, we define stretch-width and prove \cref{thm:stw-gen-cw,thm:separation}.
In~\cref{sec:matrices}, we show \cref{thm:matrix-characterization,thm:approx-stw}.
  In~\cref{sec:overlap}, we prove \cref{lem:weakly-sparse-overlap}.
  In~\cref{sec:subd}, we establish~\cref{thm:subd}.
  Finally in~\cref{sec:tw-bound}, we show \cref{thm:tw-bound} and draw the algorithmic consequences of~\cref{prop:subexp,cor:ecml}.

\section{Stretch-width}\label{sec:stw}

An ordered graph is a pair $(G, \prec)$ where $G$ is a graph and $\prec$ a strict total order on $V(G)$.
We write $u \preceq v$ whenever $u \prec v$ or $u = v$.  
Let $(G, \prec)$ is an ordered graph, and $X \subseteq V(G)$.
We now define some objects depending on $\prec$, but as the order will be clear from the context, we omit it from the corresponding notations.

The minimum and maximum of $X$ along $\prec$ are denoted by $\min(X)$ and $\max(X)$, respectively.
The \emph{convex closure} or \emph{span} of $X$ is $\conv{X} := \{v \in V(G)~|~\min(X) \preceq v \preceq \max(X)\}$.
Two sets $X,Y \subseteq V(G)$ are \emph{in conflict}\footnote{In a similar context in~\cite{twin-width4}, the verb \emph{overlap} was also used. In this paper, we will reserve \emph{overlap} for intersecting intervals (actually edges) that are not nested, notion which we will later use.}, or $X$ \emph{\crosses}~$Y$, if $\conv{X} \cap \conv{Y} \neq \emptyset$.
Note that this does not imply that $X$ and~$Y$ themselves intersect, and indeed we will mostly use this notion for two disjoint sets~$X,Y$.

Let now $\P$ be a~partition of $V(G)$, $\R(\P)$ its red graph, and $X \in \P$.
We say that $Y \in \P \setminus \{X\}$ \emph{\interfs} with $X$ if $Y$ \crosses $N_{\R(\P)}[X]$.
Note that it may well be that $Y$ \interfs with $X$, but not vice versa.
The \emph{\stret} of the part $X \in \P$, denoted by $\str{X}$, is then defined as the number of parts in $\P$ \interfing with $X$.
In turn, the \emph{\stret} of $\P$ is the maximum over every part $Z \in \P$ of~$\str{Z}$. 
The \emph{stretch-width} of the ordered graph $(G,\prec)$, denoted by $\stw(G,\prec)$, is the minimum, taken among every partition sequence $\P_n, \ldots, \P_1$ of~$G$, of $\max_{i \in [n]} \str{\P_i}$.
Finally the \emph{stretch-width} of $G$, denoted by $\stw(G)$, is the minimum of $\stw(G,\prec)$ taken among every total order $\prec$ on $V(G)$.

Notice the similarity with reduced bandwidth.
We also seek a~sequence without ``long'' red edges.
When witnessing low reduced bandwidth, one could use different (and incompatible) vertex orderings for the different red graphs.
To witness low stretch-width, we need a~stronger property: the existence of a~``global'' vertex ordering such that no red graph of the sequence has a~long edge along this single order.

\subsection{An example forcing interleaved parts}

The definition of the stretch-width of an unordered graph may seem somewhat contrived.
If we are to pick the order~$\prec$, why not choosing one along which we will perform the contractions (thereby making every part an interval, and simplifying greatly the definition)?
We will now see that this is in fact not always possible.
There is a~family of very simple graphs, with bounded treewidth, hence bounded stretch-width (see~\cref{sec:stww-cw}), but such that there is no vertex ordering $\prec$ that simultaneously~witnesses the low stretch-width, and invariably presents the two vertices to be contracted (or parts to be merged) consecutively.

For every positive integer $k$, let $H_k$ be the (series-parallel) graph obtained by adding $k$ internally vertex-disjoint 4-edge paths $s,a_i,b_i,c_i,t$ (for $i \in [k]$) between two fixed vertices $s$ and $t$; see~left of~\cref{fig:H(k)}.
The graph $H_k$ has $3k+2$ vertices and treewidth at most~2.
We will see in the next section that the much more general classes of bounded clique-width have bounded stretch-width.
But let us give a~direct argument for $H_k$, to get familiar with this new width parameter.

We choose the order $s \prec a_1 \prec b_1 \prec c_1 \prec a_2 \prec b_2 \prec c_2 \prec \ldots a_k \prec b_k \prec c_k \prec t$; see~right of~\cref{fig:H(k)}.
For the partition sequence, we will maintain three parts $A, B, C$ such that all the other parts are singletons.
Initially, we have $A=\{a_1\}$, $B=\{b_1\}$, and $C=\{c_1\}$.
Then, for $i$ going from 2 to $k$, we merge $A$ with $\{a_i\}$, then $B$ with $\{b_i\}$, and $C$ with $\{c_i\}$.
In the figure, the parts $A, B, C$ are represented after the iteration $i=3$.
When there are only five parts left (namely $\{s\}, A, B, C, \{t\}$), we finish the sequence in any fashion.

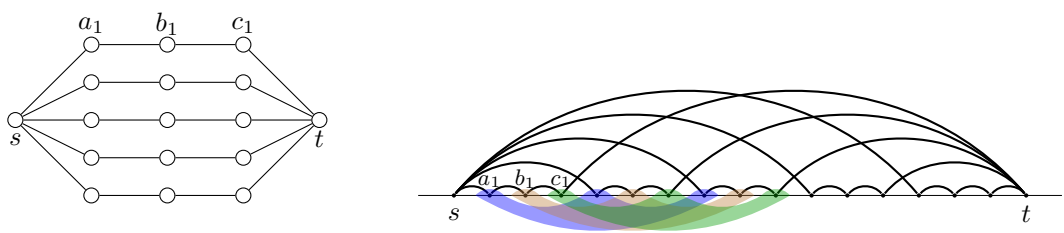
\begin{figure}[h!]
  \centering
  \begin{tikzpicture}[vertex/.style={fill,circle,inner sep=0.02cm},lvertex/.style={draw,circle,inner sep=0.07cm}]
    \def\k{5}
    \def\s{0.5}
    %H(\k)
    \node[lvertex] (a) at (0,\s * \k/2 + \s/2) {} ;
    \node[lvertex] (b) at (4,\s * \k/2 + \s/2) {} ;
    \foreach \i in {1,...,\k}{
      \foreach \j in {1,2,3}{
        \node[lvertex] (a\i\j) at (\j,\s * \i) {} ;
      }
      \draw (a) -- (a\i1) -- (a\i2) -- (a\i3) -- (b) ;
    }
    \foreach \i/\j/\p in {0/{\s * \k/2 + \s/2 - 0.25}/$s$,4/{\s * \k/2 + \s/2 - 0.25}/$t$,1/{\s * \k+0.25}/{$a_1$},2/{\s * \k+0.28}/{$b_1$},3/{\s * \k+0.25}/{$c_1$}}{
      \node at (\i,\j) {\p} ;
    }

    \begin{scope}[xshift=5.3cm, yshift=0.5cm]
    \pgfmathtruncatemacro\n{3*\k+2} % number of vertices
    \def\l{8} % length of the representation

    % vertices of (G,<)
    \draw[very thin] (0,0) -- (\l+\l/\n,0) ;
    \foreach \i in {1,...,\n}{
      \node[vertex] (\i) at (\l * \i/\n,0) {} ;
    }
    \foreach \i/\j/\p in {{\l / \n}/{-0.25}/$s$,{\l}/{-0.25}/$t$, {2 * \l / \n}/{0.19}/{\footnotesize{$a_1$}},{3 * \l / \n}/{0.22}/{\footnotesize{$b_1$}},{4 * \l / \n}/{0.19}/{\footnotesize{$c_1$}}}{
      \node at (\i,\j) {\p} ;
    }
    
    % edges of (G,<)
    \foreach \i/\j/\b in {2/3/50,3/4/50, 5/6/50,6/7/50, 8/9/50,9/10/50, 11/12/50,12/13/50, 14/15/50,15/16/50,
      1/2/50, 1/5/50, 1/8/50, 1/11/50, 1/14/50,
      16/17/50, 13/17/50, 10/17/50, 7/17/50, 4/17/50}{
      \draw[thick] (\i) to [bend left=\b] (\j) ;
    }
    % on partition of the sequence
    \foreach \i/\c in {0/blue, 1/brown, 2/{green!60!black}}{
        \fill[opacity=0.4,\c] (1.6 * \l/\n + \i * \l/\n,0) to [bend left=50] (2.4 * \l/\n + \i * \l/\n,0) to [bend left=-30] (4.6 * \l/\n + \i * \l/\n,0) to [bend left=50] (5.4 * \l/\n + \i * \l/\n,0) to [bend left=-30] (7.6 * \l/\n + \i * \l/\n,0) to [bend left=50] (8.4 * \l/\n + \i * \l/\n,0) to [bend left=30] (1.6 * \l/\n + \i * \l/\n,0) ;
    }
    \end{scope}
  \end{tikzpicture}
  \caption{Left: The graph $H_5$. Right: The same graph drawn along an order able (together with an appropriate sequence) to witness low stretch-width. The parts $A, B, C$ maintained by the partition sequence are depicted in blue, brown, green, after the iteration $i=3$.}
  \label{fig:H(k)}
\end{figure}

Importantly, the vertices within $A$, $B$, or $C$ have the same neighborhood in $\{s,t\}$.
Thus the long black edges incident to $s$ and $t$ never become long red edges.
Note that in the midst of the $i$-th iteration (also at its start and end), every part $P$ of the current partition $\P$ has a~closed red neighborhood contained in $Z := A \cup B \cup C \cup \{a_i, b_i, c_i\}$.
Set $Z$ is an interval along $\prec$, and intersects at most~six parts of $\P$ among $A$, $B$, $C$, $\{a_i\}$, $\{b_i\}$, $\{c_i\}$.
Hence at most these six parts can \cross $Z$, and thus \interf with $P$.
This implies that $\stw(H_k,\prec)$, hence $\stw(H_k)$, is bounded by a~constant (with greater care, one can show the upper bound of~3).

Now a~vertex ordering such that the sequence can be done without \confing parts forces the $a_i$'s, the $b_i$'s, and the $c_i$'s to be essentially consecutive.
Any such attempt, like the one depicted in~\cref{fig:H(k)-wrong-order} with $s \prec a_1 \prec \ldots \prec a_k \prec b_k \prec \ldots \prec b_1 \prec c_1 \prec \ldots \prec c_k \prec t$, creates a~structure which, as~we will show in~\cref{sec:overlap}, entails large stretch-width: a large ``biclique'' of overlapping edges spanning vertices of bounded degree.  

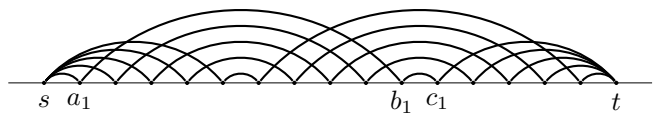
\begin{figure}[h!]
  \centering
  \begin{tikzpicture}[vertex/.style={fill,circle,inner sep=0.02cm},lvertex/.style={draw,circle,inner sep=0.07cm}]
    \def\k{5}
    \pgfmathtruncatemacro\n{3*\k+2} % number of vertices
    \def\l{8} % length of the representation

    % vertices of (G,<)
    \draw[very thin] (0,0) -- (\l+\l/\n,0) ;
    \foreach \i in {1,...,\n}{
      \node[vertex] (\i) at (\l * \i/\n,0) {} ;
    }
      \foreach \i/\j/\p in {{\l / \n}/{-0.25}/$s$,{\l}/{-0.25}/$t$, {2 * \l / \n}/{-0.25}/$a_1$,{11 * \l / \n}/{-0.25}/$b_1$,{12 * \l / \n}/{-0.25}/$c_1$}{
      \node at (\i,\j) {\p} ;
    }
    % edges of (G,<)
    \foreach \i/\j/\b in {2/11/50,3/10/50,4/9/50,5/8/50,6/7/50, 7/16/50,8/15/50,9/14/50,10/13/50,11/12/50,
      1/2/50,1/3/50,1/4/50,1/5/50,1/6/50,
      12/17/50,13/17/50,14/17/50,15/17/50,16/17/50}{
      \draw[thick] (\i) to [bend left=\b] (\j) ;
    }

  \end{tikzpicture}
  \caption{A~failed attempt at contracting along the chosen order. Observe that the first contraction involving some $b_i$ (in the ``middle'' of the order) necessarily creates a~long red edge.}
  \label{fig:H(k)-wrong-order}
\end{figure}

\subsection{Graphs of bounded clique-width have bounded stretch-width}\label{sec:stww-cw}

Clique-width and component twin-width are \emph{functionally equivalent}, that is, a~graph class has bounded clique-width if and only if it has bounded component twin-width~\cite{twin-width1,twin-width6}.
More quantitatively, it can be observed that, for every graph $G$, $\cw(G) \leq \ctww(G)+1 \leq 2 \cdot \cw(G)$; see~\cite{BarilSlides}.
We show that the stretch-width of a graph is at most its component twin-width.

\begin{reptheorem}{thm:stw-gen-cw}\label{from cw to stw}
  For every graph $G$, $\stw(G) \leq \ctww(G)-1$.
  Hence, $\stw(G) \leq 2(\cw(G) - 1)$.
\end{reptheorem}
\begin{proof}
  Let $\P_n, \ldots, \P_1$ be a~partition sequence of a~graph $G$ such that every red graph $\R(\P_i)$ has all its connected components of size at most $t := \ctww(G)$.
  We define a~total order $\prec$ on $V(G)$ as the last order of a~sequence of partial orders $\prec_i$ on $\P_i$ for~$i$ going from 1 to~$n$.
  That is, $\prec_n$ is a~total order, and we imply set $\prec~:=~\prec_n$.
  We maintain the following invariants:
  \begin{compactitem}
  \item Parts $X, Y \in \P_i$ are comparable with respect to $\prec_i$ if and only if they are in distinct connected components of~$\R(\P_i)$, and
  \item \emph{$\prec_i$ is a~total order on the connected components of $\R(\P_i)$}, that is, for every distinct connected components $C, C'$ of $\R(\P_i)$, if $X \prec_i Y$ for some $X \in C$ and $Y \in C'$, then $X \prec_i Y$ holds for every $X \in C$ and $Y \in C'$.
  \end{compactitem}

We define $\prec_1$ as the empty relation, which satisfies the invariant since $\P_n$ has a single part, $V(G)$.
We define $\prec_{i+1}$ from $\prec_i$ in the following way.
Let $X, Y$ be the two parts of $\P_{i+1}$ being merged to $Z=X \cup Y \in \P_i$.
Let $C$ be the connected component of $Z$ in $\R(\P_i)$.
Let $C_1, \ldots, C_h$ be the connected components of $C \setminus \{Z\} \cup \{X,Y\}$ in $\R(\P_{i+1})$.
We then set $P \prec_{i+1} Q$ whenever either $P \prec_i Q$ or $P \in C_a$ and $Q \in C_b$ for some pair $a < b \in [h]$. 
One can check that $\prec_{i+1}$ is a partial order satisfying the invariants.
%Note also that, if $P \prec_i Q$, then for every $j \geqslant i$, and $P', Q' \in \P_j$ such that $P' \subseteq P$ and $Q' \subseteq Q$, it holds that $P' \prec_j Q'$.

As the connected components of $\R(\P_n)$ are singletons, $\prec_n$ is a total order.
We now use~$\prec$ as a~witness of low stretch-width for the same partition sequence $\P_n, \ldots, \P_1$.
Let $X$ be any part of any partition $\P_i$, and let $C$ be the connected component of $\R(\P_i)$ containing $X$.
By the second invariant, no part of $C$ can cross a~part of $\P_i \setminus C$.
Also, by definition, $X$ may only be inhomogeneous to parts of $C$.
Therefore, the only parts that can \interf with $X$ are in~$C$; thus $\str{X} \leqslant |C|-1 \leqslant t-1$.
Finally, $\stw(G) \leqslant \stw(G,\prec) \leqslant t-1=\ctww(G)-1$. 
\end{proof}

\subsection{Separating construction}

We present a graph family with bounded stretch-width but unbounded clique-width.
Let $b \geqslant 2$ and $h \geqslant 1$ be two integers.
We build a~graph $\at{b}{h}$ on vertex set $[b^h]$.

Informally, $\at{2}{h}$ is built by first adding the paths $1,3,5,\ldots$ on ``odd'' vertices, and $2,4,6,\ldots$ on ``even'' vertices.
Then identifying the pairs $2i-1, 2i$, renaming them $i$, and adding the ``odd'' and ``even'' paths (that is, the edge between the first and the third vertices at this step corresponds to the complete adjacency between $\{1,2\}$ and $\{5,6\}$, in terms of original vertices).
And iterating this process until it runs out of vertices; see~\cref{fig:A2(5)}.

\begin{figure}[h!]
  \centering
  \resizebox{390pt}{!}{
\begin{tikzpicture}[vertex/.style={draw,circle,inner sep=0.02cm}]
  \def\z{0.5}
  \def\b{65}
  \def\t{5}
  \pgfmathtruncatemacro\s{2^\t}
  \pgfmathtruncatemacro\sm{\s-1}
  \pgfmathtruncatemacro\smm{\s-2}

  \foreach \i in {1,...,\s}{
    \node[vertex] (\i) at (\i * \z,0) {} ;
  }

  \foreach \i [count=\hh from 1] in {1,3,...,\smm}{
    \pgfmathtruncatemacro\h{2*\hh+1}
    \pgfmathtruncatemacro\ip{\i+1}
    \pgfmathtruncatemacro\hp{\h+1}
    \draw[very thick,blue] (\i) to [bend left=\b] (\h) ;
    \draw[very thick,blue] (\ip) to [bend left=-\b] (\hp) ;
    \node[draw,rounded corners,thick,blue!75!red,inner sep=0.1cm,fit=(\i) (\ip)] (a\hh) {} ;
  }
  \node[draw,rounded corners,thick,blue!75!red,inner sep=0.1cm,fit=(31) (32)] (a16) {} ;

  \foreach \i [count=\hh from 1] in {1,3,...,14}{
    \pgfmathtruncatemacro\h{2*\hh+1}
    \pgfmathtruncatemacro\ip{\i+1}
    \pgfmathtruncatemacro\hp{\h+1}
    \draw[very thick,blue!75!red] (a\i) to [bend left=\b] (a\h) ;
    \draw[very thick,blue!75!red] (a\ip) to [bend left=-\b] (a\hp) ;
    \node[draw,rounded corners,thick, blue!50!red,inner sep=0.03cm,fit=(a\i) (a\ip)] (a-2-\hh) {} ;
  }
  \node[draw,rounded corners,thick,blue!50!red,inner sep=0.03cm,fit=(a15) (a16)] (a-2-8) {} ;

  \foreach \i [count=\hh from 1] in {1,3,...,6}{
    \pgfmathtruncatemacro\h{2*\hh+1}
    \pgfmathtruncatemacro\ip{\i+1}
    \pgfmathtruncatemacro\hp{\h+1}
    \draw[very thick,blue!50!red] (a-2-\i) to [bend left=\b] (a-2-\h) ;
    \draw[very thick,blue!50!red] (a-2-\ip) to [bend left=-\b] (a-2-\hp) ;
    \node[draw,rounded corners,thick,blue!25!red,inner sep=0.03cm,fit=(a-2-\i) (a-2-\ip)] (a-3-\hh) {} ;
  }
  \node[draw,rounded corners,thick, blue!25!red,inner sep=0.03cm,fit=(a-2-7) (a-2-8)] (a-3-4) {} ;

  \draw[very thick,blue!25!red] (a-3-1) to [bend left=\b] (a-3-3) ;
  \draw[very thick,blue!25!red] (a-3-2) to [bend left=-\b] (a-3-4) ;
  
\end{tikzpicture}
}
\caption{The graph $\at{2}{5}$, the edge colors correspond to a~step in the recursive construction.}
\label{fig:A2(5)}
\end{figure}
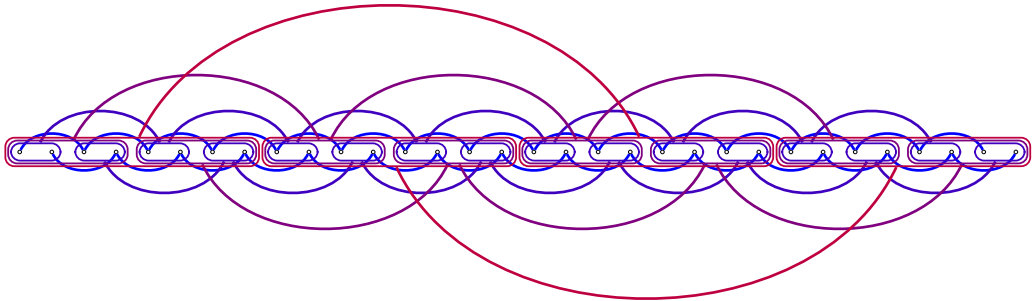

We give a~general definition, but will only use $\at{3}{h}$.
Let $\tree{b}{h}$ the complete $b$-ary tree of depth $h$.
Think of the vertices of $\at{b}{h}$ as the leaves of $\tree{b}{h}$, named from left to right $1, 2, \ldots, b^h$.
The level of a node of $\tree{b}{h}$ is $h$ minus the depth of the node. For example the level of a leaf is $0$.
We call $\Root{l}{i}$ for $i$-th internal node (from left to right) on the $l$-th level.
At every level $l$ of $\tree{b}{h}$, we add $b$ node-disjoint paths $P_1, \dots, P_b$ such that for each $i$ of $[0, b-1]$, $P_i = \Root{l}{i}, \Root{l}{i+b}, \Root{l}{i + 2b}, \ldots$; see \Cref{fig:A3(4)}.
When two nodes at level $l$ are linked by an edge of such a path, we say that they are \emph{$l$-linked}.

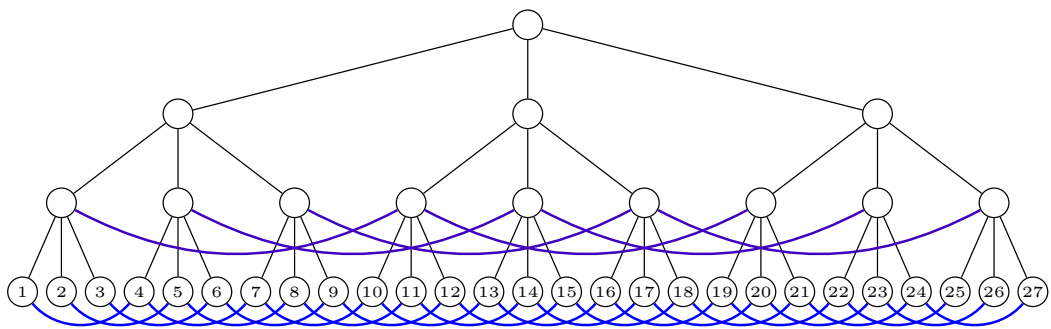
\begin{figure}
  \resizebox{400pt}{!}{
  \begin{tikzpicture}
    \def\b{3} %tree arity
    \def\s{1}     %veritcal separation 
    \def\sh{1.3}  %horizontal separation
    \def\z{4}     %depth of the tree
      \pgfmathtruncatemacro\zm{\z-1}
      \pgfmathtruncatemacro\zmm{\zm-1} 

    %nodes of the tree: aij is the j-th node at depth i   
    \foreach \i in {1,...,\z}{
      \pgfmathtruncatemacro\k{\b^\i/\b} 
      \foreach \j in {1,...,\k}{
        \pgfmathsetmacro\h{\b^\z/(\b * \b^\i)}  
        \node[draw,circle] (a\i\j) at (\j * \h * \sh - \h * \sh / 2, - \i * \s) {} ;
      }
    }
    \pgfmathtruncatemacro\k{\b^\z/\b}
    \foreach \j in {1,...,\k}{
        \pgfmathsetmacro\h{1/\b}  
        \node at (\j * \h * \sh - \h * \sh / 2, - \z * \s) {\tiny{$\j$}} ;
    }

    %edges of the b-ary tree
    \foreach \i [count = \ip from 2] in {1,...,\zm}{
      \pgfmathtruncatemacro\k{\b^\i/\b}
      \foreach \j in {1,...,\k}{
        \pgfmathtruncatemacro\jmm{\b * \j - 2}
        \pgfmathtruncatemacro\jm{\b * \j - 1}
        \pgfmathtruncatemacro\jp{\b * \j}
        \draw (a\i\j) -- (a\ip\jmm) ;
        \draw (a\i\j) -- (a\ip\jm) ;
        \draw (a\i\j) -- (a\ip\jp) ;
        }
    }

    %edges of the graph
    \foreach \i in {3,...,\z}{
      \pgfmathtruncatemacro\f{100 * \i/\z}
      \pgfmathtruncatemacro\be{-50 * \i/\z * \i/\z}
      \def\c{blue!\f!red}
     \pgfmathtruncatemacro\km{\b^\i/\b - 3} 
     \foreach \j [count = \jp from 4] in {1,...,\km}{
       \draw[thick,\c] (a\i\j) to [bend left=\be] (a\i\jp) ;
     }
     }
  \end{tikzpicture}
  }
  \caption{Tree representation of $\at{3}{4}$. The vertex set of $\at{3}{4}$ is made by the leaves, and a~colored edge $uv$ (one not part of the tree) is a~biclique between the two disjoint sets of leaves of the two subtrees rooted at $u$ and $v$.}
  \label{fig:A3(4)}
\end{figure}

From the tree $\tree{b}{h}$, we build $\at{b}{h}$ as follows: the vertices of $\at{b}{h}$ are the leaves of $\tree{b}{h}$ and two vertices $u$ and $v$ are adjacent in $\at{b}{h}$ if and only if there exists a level $l$, an ancestor $\Root{l}{i}$ of $u$ and an ancestor $\Root{l}{j}$ of $v$ such that $\Root{l}{i}$ and $\Root{l}{j}$ are $l$-linked, that is, $|j - i| = b$.

For a given node $\Root{l}{i}$ of $\tree{b}{h}$, we denote by $V(\Root{l}{i})$ the set of leaves (hence, vertices of $\at{b}{h}$) that are descendant of $\Root{l}{i}$.
Observe that $\at{b}{h}[V(\Root{l}{i})]$ is isomorphic to $\at{b}{l}$.
If $S$ is a set of vertices of $\at{b}{h}$ then the \emph{level of $S$} is defined as the smallest level $l$ for which there exists at most two consecutive nodes $\Root{l}{i}, \Root{l}{i+1}$ on the $l$-th level of $\tree{b}{h}$ such that every element of $S$ is a~descendant of $\Root{l}{i}$ or $\Root{l}{i+1}$.

We can directly establish the following bound on the stretch-width of $\at{3}{h}$:
\begin{lemma}\label{stw-at}
For every integer $h$, the stretch-width of $\atG{3}{h}$ is at most 9.
\end{lemma}
\begin{proof}
Let $h$ be a positive integer.
Let $\prec$ be the left-right order on the leaves of $\tree{b}{h}$, and for each $l$ in $[h]$, let $\mathcal{P}_l$ be the partition $\{V(\Root{l}{1}), V(\Root{l}{2}), \ldots \}$.

Observe that $V(\Root{l}{i})$ is not homogeneous to $V(\Root{l}{j})$ if and only if $|j-i| = 1$. As for any $i \neq j$ $\conv{V(\Root{l}{i}} \cap \conv{V(\Root{l}{j}}$ is empty, the stretch of $\mathcal{P}_l$ along $\prec$ is at~most $2$ for every $l$.
In addition, observe that one can go from $\mathcal{P}_l$ to $\mathcal{P}_{l+1}$ by merging only consecutive parts. As each part of $\mathcal{P}_{l+1}$ is composed of exactly 3 consecutive parts of $\mathcal{P}_l$, the stretch of the partitions between $\mathcal{P}_l$ and $\mathcal{P}_{l+1}$ is at most $3 \cdot 3 = 9$.
\end{proof}

\begin{lemma}\label{bound on degree}
If $v$ is a vertex of $\at{3}{h}$ then the degree $d$ of $v$ verifies
$\frac{3^{h-1} - 1}{2} \leqslant d \leqslant 3^{h-1} - 1.$
\end{lemma}
\begin{proof}
Let $v$ a vertex of $\at{3}{h}$.
For a vertex $w$ to be adjacent to $v$, it is necessary that there exists a level $l$ and two $l$-linked nodes $\Root{l}{i}$, $\Root{l}{j}$ that are ancestors of $v$ and $w$, respectively.
Observe that such a level is unique.
Hence, the ancestor of $v$ on the $l$-th level for $l < h-1$ yields either $3^l$ or $2 \cdot 3^l$ neighbors of $v$, depending on whether it is incident to one or two $l$-linked nodes in $\at{b}{h}$.
In addition, there is no $l$-linked nodes for $l \geqslant h-1$, thus ancestors at level $h$ and $h-1$ do not contribute to any edge of~$\at{b}{h}$.

As $\sum_{i \leq h-2} 3^l = \frac{3^{h-1} - 1}{2}$, the degree of $v$ is lowerbounded by $\frac{3^{h-1} - 1}{2}$ and upperbounded by $3^{h-1} - 1$.
\end{proof}

\begin{lemma}\label{degree-type}
For every level $l$, for every two integers $i \neq j$, the number of neighbors of a~vertex $v \in V(\Root{l}{i})$ that are in $V(\Root{l}{j})$ is:
\begin{compactenum}
	\item in the interval $[3^{l-1}, \frac{3^{l} - 1}{2}]$ when $|j - i| = 1$, and
	\item equal to $3^l$ or $0$ when $|j - i| > 1$.
\end{compactenum}
\end{lemma}
\begin{proof}
Let $v$ a vertex of $\at{3}{h}$ such that $v$ is in $V(\Root{i}{l})$ for a level $l$.
Let $j$ be such that $|j - i| = 1$.
Then if $\Root{l-1}{i'}$ is the ancestor of $v$ at the $(l-1)$-st level, there exists a child $\Root{l-1}{j'}$ of $V(\Root{l}{j})$ which is $(l-1)$-linked to $a_v$.
Thus the number of neighbors of $v$ that are in $V(\Root{l}{j})$ is at least $3^{l-1}$.
Observe that every ancestor of $v$ at level $k < l$ is $k$-linked to at most one descendant of $\Root{l}{j}$.
In addition, for every integer $k \geqslant l$, the ancestor at level $k$ of $\Root{l}{i}$ and the one at level $k$ of $\Root{l}{j}$ are not $k$-linked.
Hence the maximum number of neighbors of $v$ that are in $V(\Root{l}{j})$ is upperbounded by $\sum_{k < l} 3^k = \frac{3^{l} - 1}{2}$.

Assume now that $|j - i|$ is at least $2$.
Then, none of the ancestors of $v$ at level at most $k < l$ is $k$-linked to a descendant of $\Root{l}{j}$.
Thus, either an ancestor of $v$ at level $k \geqslant l$ is $k$-linked to a (non-strict) ancestor of $\Root{j}{l}$, in which case $v$ is fully adjacent to $V(\Root{l}{j})$, or none of the ancestors of $v$ is linked to an ancestor of $\Root{l}{j}$, in which case $v$ has no neighbors in $V(\Root{l}{j})$.
\end{proof}

Remember that the level of a subset~$S$ of the vertices of~$\at{3}{h}$ is the smallest level~$l$ for which there is at most two consecutive nodes $\Root{l}{i}, \Root{l}{i+1}$ at level $l$ such that every vertex of $l$ is either a descendant of $\Root{l}{i}$ or of $\Root{l}{i+1}$.

We will now prove that for any partition sequence $\mathcal{P}_1, \mathcal{P}_2, \ldots$ of $\at{3}{h}$, the size of a~largest red component among every partition is lowerbounded by a function of $h$.
To do so, we show that as long as there is no large part in the sequence, any part of the sequence has a small level, and thus is ``localized'' in the graph.
We then prove that, given a large part $S$ in one of the partitions, this part needs to have a red neighbor the size of which is at least linear in the size of $S$.
We finally build a large red component based on these two facts.
Informally, we take the largest part of a suitable partition.
This part is localized, say on the left side, and we use the second fact to build a smaller red neighbor to the right of this part, and continue this process until we reach a~part of constant size.

\begin{lemma}\label{localisation}
  Let $\mathcal{P}_1, \mathcal{P}_2, \ldots$ be a~partition sequence of $\at{3}{h}$ with maximum red degree at most $d$.
  Let, for any integer $t$, $s(t)$ be the smallest integer such that one of the parts of $\mathcal{P}_{s(t)}$ has size at least $t$.
  Then for every positive integer~$t$, for every part $Q \in \mathcal{P}_{s(t)}$, the level of $Q$ is at most ${\log_3(t) + \log_3(d) + 5}$.
\end{lemma}

\begin{proof}
  %Let $\mathcal{P}_1, \mathcal{P}_2, \ldots$ be a~partition sequence of~$\at{3}{h}$ with red degree at most $d$.
  %Let, for any integer $t$, $s(t)$ be the smallest integer such that one of the parts of $\mathcal{P}_{s(t)}$ has size at least $t$.
%Let $t$ be a positive integer.
Each part of $\mathcal{P}_{s(t)}$ contains at most $2t$ vertices.
Let $Q$ be~a part of $\mathcal{P}_{s(t)}$ and $l$ be its level.
As $l-1$ is strictly smaller than the level of $Q$, there are two integers $i, j$ with $i+1 < j$ such that $Q \cap V(\Root{l-1}{i})$ and $Q \cap V(\Root{l-1}{j})$ are both non-empty.
Going one level further down, the intermediate set $V(\Root{l-1}{i+1})$ splits into three parts.
Thus, there are two integers $i'$, $j'$ such that $i' + 3 < j'$ and $Q$ has a vertex $u$ in $V(\Root{l-2}{i'})$, and a vertex $v$ in $V(\Root{l-2}{j'})$.

By \Cref{degree-type}, as $|j' - i'| > 1$, $u$ has either $3^{l-2}$ or $0$ neighbors in $V(\Root{l-2}{j'})$, while by \Cref{bound on degree}, $v$ has between $\frac{3^{l-3}-1}{2}$ and $3^{l-3} - 1$ neighbors in $V(\Root{l-2}{j'})$.
Thus the symmetric difference of the neighborhoods of $u$ and $v$, say $A$, contains at least $\frac{3^{l-3}-1}{2}$ vertices.
The maximum red degree of $Q$ is at~most $d$, and for each part $T \neq Q$ of $\mathcal{P}_{s(t)}$ containing a vertex of $A$, $Q$ has a red edge toward $T$.
Thus $A$ is contained in at most $d+1$ parts of $\mathcal{P}_{s(t)}$.
Hence $|A| \leqslant 2t\cdot (d+1)$.
Since $|A| \geqslant \frac{3^{l-3}-1}{2}$, we get $\frac{3^{l-3}-1}{2} \leqslant 2t \cdot (d+1)$, thus $l-3 \leqslant \log_3(4t \cdot (d+1))$.
\end{proof}

\begin{lemma}\label{right-large-red-neigh}
Let $\mathcal{P}$ any partition of $\at{3}{h}$ within a~partition sequence of maximum red degree at~most $d$.
Let $S$ a part in $\mathcal{P}$, at level $l$, such that there is at least $3^l$ vertices in $\at{3}{h}$ strictly to the right of $S$.
Then, there is a part $T$ of $\mathcal{P}$ in the red neighborhood of $S$ with $|T| \geqslant 3^{l-1}/d$, and such that $T$ has a vertex strictly to the right of $S$.
Furthermore, if $S$ is contained in $V(\Root{l}{i}) \cup V(\Root{l}{i+1})$ (resp. only $V(\Root{l}{i})$), then $T$ contains a vertex $v$ in $V(\Root{l}{i+2})$ (resp. $V(\Root{l}{i+1})$).
\end{lemma}
\begin{proof}
Let $l$ be the level of $S$.
There are two integers $i+1 < j$ such that $S \cap V(\Root{l-1}{i}$ contains a vertex $u$, and  $S \cap V(\Root{l-1}{j}$ contains a vertex $v$.
As there is at least $3^{l}$ vertices strictly to the left of $S$, $\Root{l-1}{j+1}$ is well defined.
By \cref{degree-type}, the number of neighbors of $u$ within $V(\Root{l-1}{j+1})$ is either $0$ or $3^{l-1}$, and the number of neighbors of $v$ within $V(\Root{l-1}{j+1})$ is between $3^{l-1}$ and $\frac{3^{l} - 1}{2}$.
Thus the symmetric difference of neighborhoods of $u$ and $v$ inside $V(\Root{l-1}{j+1})$, say $A$, contains at least $3^{l-1}$ vertices.
As each part of $\mathcal{P}$ containing a vertex of $A$ is inhomogeneous to $S$, there is a~part $T \in \mathcal{P}$ containing at least $3^{l-1}/d$ vertices of $A$.
\end{proof}

\begin{theorem}\label{large-ctww-at}
The component twin-width of $\at{3}{h}$ is $\Omega(h/\log h)$.
\end{theorem}
\begin{proof}
Consider a~partition sequence $\mathcal{P}_1, \mathcal{P}_2, \ldots$ of $\at{3}{h}$ of maximum red degree at most $d$.
Let $t = 3^{h/2}$, and let $s$ be the smallest integer such that $\mathcal{P}_s$ contains a part of size at least $t$.
We will show that the red component of $\mathcal{P}_s$ containing its largest part is large.

Consider $P$ the only part of $\mathcal{P}_s$ of size at least $t$.
Let $l = h/2 + \log_3(d) + 5$.
By \Cref{localisation}, the level of $P$ is at most $l$.
Thus, up to symmetry, there is at least $3^h/3$ vertices strictly to the right of $P$.
We consider the sequence of parts $(P_k)_k$ such that $P_0 = P$, and $P_{k+1}$ is the part obtained by the application of \Cref{right-large-red-neigh} on $P_k$, as long as the conditions of \Cref{right-large-red-neigh} are satisfied.
We state the three following facts on the sequence $(P_k)_k$.
\begin{compactenum}
	\item $P_k$ is different from $P_i$ for every  $i < k$.
	\item $|P_k| \geqslant 3^{\text{level}(P_{k-1}) - 1}/d \geqslant |P_{k-1}|/(3d)$.
	\item There is a finite sequence of positive integers $(i_k)_k$ such that for any $m$, $P_m$ is contained in $V(\Root{l}{i_m}) \cup V(\Root{l}{i_m + 1})$, and $i_m \leqslant i_{m+1} \leqslant i_m + 2$.
\end{compactenum}
Indeed for every $k$, $P_{k+1}$ always has a vertex strictly to the right of $P_k$, thus Item 1 holds.
Item 2 directly follows from the bound on the size of the part of \Cref{right-large-red-neigh}.
The lower bound of Item 3 comes from the fact that $P_m$ has a vertex strictly to the right of $P_{m-1}$, and the upper bound from the fact that $P_m$ contains a vertex in $V(\Root{l}{i_{m-1} + 1})$. 
 
Consider the largest integer $n$ for which $P_n$ is defined.
Integer $n$ satisfies $3^{\text{level}(P_n) - 1} > d$, or there is less than $3^{\text{level}(P_n)}$ (thus, less than $3^l$) vertices strictly to the right of $P_n$.
As the level of $P_n$ is at least $\log_3 |P_n|$, if the former condition is satisfied, $|P_n| \leqslant 3d$.
But Item 2 ensures that $|P_n|$ is at least $P_0 / (3d)^n$, thus $n \geqslant \log_{3d} t$.
If the latter condition is satisfied, then $i_n$ is at least $(1/3^{h/2}) \cdot (3^h/3) - 1 = 3^{h/2 - 1} - 1$, and thus Item 3 ensures that $n$ is at least $\frac{3^{h/2 - 1} - 1}{2}$.

Hence, as Item 1 ensures that the red component containing $P$ is of size at least $n$, this component is of size at least $\log_{3d} t$.
Let $x$ be the component twin-width of $\at{3}{h}$.
Then in a sequence witnessing that fact, the size of a largest red component is at most $x$, thus the maximum red degree is at most $x$ (even $x-1$).
Therefore $\log_{3x} t \leqslant x$, and so $\log t \leqslant x \log (3x)$.
Thus $3x \geqslant \frac{h}{2}/\log{\frac{h}{2}}$, and $x=\Omega(h/\log h)$.
\end{proof}

\cref{thm:separation} is a~consequence of~\cref{stw-at,large-ctww-at}. 

\section{Matrix characterization}\label{sec:matrices}

Let us first reinterpret the definition of stretch-width on symmetric (ordered) matrices.
A~symmetric \emph{partition} of a~symmetric matrix $M$ is a~pair $(\R,\CC)$ such that $\R$ is a partition of the row set of $M$, $\rows(M)$, $\CC$ is a~partition of the column set, $\cols(M)$, and $\CC$ is symmetric to $\R$, i.e., two rows $r_i$ and $r_j$ are in the same part if and only if the symmetric columns $c_i$ and $c_j$ are in the same part.
Hence each row part corresponds to a (unique) symmetric column part.

A~(symmetric) \emph{division} of a~symmetric matrix $M$ is a (symmetric) partition of $M$ every row (resp.~column) part of which is on consecutive rows (resp.~columns).
Given a~row part $R \in \R$, and a~column part $C \in \CC$, the~\emph{zone $R \cap C$} of $M$ is the submatrix of $M$ with row set~$R$ and column set~$C$.
A~zone $R \cap C$ is \emph{diagonal} if $R$ and $C$ are symmetric parts.
A~zone is \emph{non-constant} if it contains two distinct entries.
A~symmetric \emph{partition sequence} of an $n \times n$ $0,1$-matrix $M$ is a sequence $(\R_n,\CC_n), \ldots, (\R_1,\CC_1)$ where $(\R_n,\CC_n)$ is the \emph{finest partition} (with $n$ row parts and $n$ column parts), $(\R_1,\CC_1)$ is the coarsest partition (with one row part and one column part), and for every $i \in [2,n]$, $(\R_{i-1},\CC_{i-1})$ is obtained from $(\R_i,\CC_i)$ by merging together two row parts, and the symmetric two column parts.

So far, we were following the definitions of~\cite{twin-width1,twin-width2} (in the symmetric case).
Instead of defining the \emph{error value} which leads to the twin-width of a~matrix, we introduce the \emph{stretch value}.
The \emph{stretch value} of a~row part $R$ of a~matrix partition $(\R,\CC)$ is the number of column parts \confing with the union of columns parts $C$ such that $C$ is the symmetric of $R$, or $R \cap C$ is non-constant.
The \emph{stretch value} of a~column part is defined symmetrically.
The \emph{stretch value} of a~partition $(\R,\CC)$ is the maximum stretch value of a~part of~$(\R,\CC)$.
Finally, the \emph{stretch-width} of a~symmetric $0,1$-matrix $M$ is the minimum among every symmetric partition sequence $\mathcal S$ of~$M$ of the maximum stretch value among partitions of~$\mathcal S$.
Observe that for any ordered graph $(G, \prec)$, the stretch-width of $(G, \prec)$ is equal to the stretch-width of its adjacency matrix.

\begin{figure}[h!]
  \centering
  \begin{tikzpicture}[scale=0.5]
    \pgfmathsetseed{1236} % Set the seed for the random number generator
    \def\n{14}
    \pgfmathtruncatemacro\nm{\n-1}
    \foreach \i in {1,...,\nm} {
      \pgfmathtruncatemacro\nj{\nm-\i+1}
        \foreach \j in {1,...,\nj} {
          \pgfmathrandominteger{\rand}{0}{1}
          \xdef\myarray{\i,\j,\rand} % Store the generated value
          \node at (\i,\j) {\rand};
          \node at (\n-\j+1,\n-\i+1) {\rand};
        }
    }
    \foreach \i in {1,...,\n}{
      \node at (\i,\n-\i+1) {\textcolor{blue}{0}} ;
    }

    \draw (0.3,0.4) -- (0.3,\n+0.6)--++(0.3,0) ;
    \draw (0.3,0.4) --++(0.3,0) ;
    
    \draw (0.7+\n,0.4) -- (0.7+\n,\n+0.6)--++(-0.3,0) ;
    \draw (0.7+\n,0.4) --++(-0.3,0) ;
    
    \foreach \i in {2.5,5.5,8.5,12.5}{
      \draw[very thick] (\i,0.5) -- (\i,\n+0.5);
      \draw[very thick] (0.5,\n-\i+1) -- (\n+0.5,\n-\i+1);
    }

    \foreach \i/\j in {4.5/$R_4$,8/$R_3$,11/$R_2$,13.5/$R_1$}{
      \node at (-0.5,\i) {\j}; 
    }
     \foreach \i/\j in {7/$C_3$}{
      \node at (\i,-0.3) {\j}; 
     }

     \fill[opacity=0.2] (5.5,9.5) -- (8.5,9.5) -- (8.5,12.5) -- (5.5,12.5) -- cycle;
  \end{tikzpicture}
  \caption{A symmetric division of a~symmetric $0,1$-matrix. $R_3$ and $C_3$ are symmetric parts. The zone $R_2 \cap C_3$ is shaded. Column part $C_3$ is not 3-wide since the deletion of $R_2, R_3, R_4$ leaves only two distinct column vectors in $C_3$, namely $(0,1,0,1)$ and $(1,0,0,0)$. It is however 2-wide since $R_1 \cap C_3$ has two distinct column vectors (and $R_1$ is too far from the diagonal zone to be removed).} 
  \label{fig:symm-div}
\end{figure}

The following is the counterpart of the so-called \emph{rich divisions}~\cite{twin-width4} tailored for stretch-width.
If $R$ is a set of rows, and $C$ is a set of columns of a~matrix~$M$, we denote by $R \setminus C$ the zone $R \cap (\cols(M) \setminus C)$, that is the submatrix formed by $R$ deprived of the columns of~$C$ (and symmetrically for $C \setminus R$). 
In a~division $(\R=(R_1, \dots, R_n), \CC=(C_1, \dots, C_m))$, a row part $R_i$ is $k$-wide if for every $k$~consecutive columns parts $C_j, \ldots, C_{j + k - 1}$ containing the symmetric of~$R_i$, $R_i \setminus \bigcup_{j \leq h \leq j+k-1} C_h$ contains at least $k$ distinct rows.
The \emph{$k$-wideness} of column parts is defined symmetrically; see~\cref{fig:symm-div}.
A~division $(\R, \CC)$ is \emph{$k$-wide} if all its row and column parts are $k$-wide.
The division is \emph{$k$-diagonal} if none of the row and column parts is $k$-wide.

Given a~set of rows (or columns) $X$ of a matrix $M$, we keep the notation $\conv{X}$ for the set of rows (or columns) of $M$ with indices between the minimum and the maximum indices of~$X$.

\begin{theorem}\label{stw bbd imp no large div}
    For every symmetric $0,1$-matrix $M$ and natural $k$, if $\stw(M) \leq k$, then $M$ has no $9k$-wide division. 
\end{theorem}
\begin{proof}
  Let $\mathcal D= (\R, \CC)$ be a~symmetric division of $M$.
  Let $\mathcal P= (\R'_n, \CC'_n), \ldots, (\R'_1, \CC'_1)$ be a~symmetric partition sequence of $M$ with stretch value at most~$k$. 
  Let $s$ be the largest integer (that is, first time within the partition sequence) for which there is a~row part~$R' \in \R'_s$ such that $\conv{R'}$ contains a~row part $R \in \R$ of the division $\mathcal D$.
  Observe that, by symmetry of $\mathcal D$ and $\mathcal P$, it happens at the same time in columns.
  We will prove that $R$ is not $9k$-wide.

  Set $\mathcal S := \left \{ T \in \mathcal{R}'_s~|~\conv{T} \cap R \neq \emptyset \right \}$.
  Note that $\mathcal{S}$ is the set of row parts of $\R'_s$ \crossing $R$, and that $R' \in \mathcal{S}$.
  As $\conv{R} \subset \conv{R'}$, every part in $\mathcal S$ \crosses $R'$.
  Thus, it should hold that $|\mathcal S| \leqslant k$ because $(\R'_s, \CC'_s)$ has stretch value at~most~$k$.

  For each $T$ in $\mathcal S$, we define $C_T:= \{c \in \cols(M)~|~c~\in~C,~C~\in~\CC'_s,~\text{and}~C~\cap~T~\text{is non-}$ $\text{constant or } C \text{ is the symmetric of } T\}$.
  For every $T \in \mathcal{S}$, $C_T$ \crosses at~most $k$~parts of~$\mathcal{C}'_s$, as the stretch value of $T$ is at~most~$k$.
  Let $C'$ be the symmetric of $R'$.
  As $T$ \crosses $R'$, the symmetric of $T$ \crosses $C'$.
  Since the symmetric of $T$ is contained in $C_T$, $C_T$ \crosses $C'$, thus $\conv{C'} \cap \conv{C_T} \neq \emptyset$.
  
  As for every $T \in \mathcal{S}$, $\conv{C'} \cap \conv{C_T} \neq \emptyset$, there is $T_1, T_2 \in \mathcal{S}$ such that
  $$\conv{C_{T_1}} \cup \conv{C_{T_2}} \cup \conv{C'} = \conv{\bigcup_{T \in \mathcal S} C_T}.$$
  Take indeed $T_1 \in \mathcal S$ such that $C_{T_1}$ contains the column of minimum index in $\bigcup_{T \in \mathcal S} C_T$, and $T_2 \in \mathcal S$ such that $C_{T_2}$ realizes the maximum index. 
  As $C_{T_1}$ and $C_{T_2}$ \cross $C'$, the parts in conflict with $C_{T_1}$, with $C'$, and with $C_{T_2}$ are consecutive, so their union contains at most $3k$ parts of $\CC'_s$.
  Thus $\bigcup_{T \in \mathcal{S}} C_T$ \crosses at most $3k$ parts of $\CC'_s$. 
 
  Observe that, except for $C'$, every part in $\CC'_s$ is covered by the union of two consecutive parts of $\CC$.
  Part $C'$ is itself covered by the union of three consecutive parts of $\CC$. %: $\conv{C'}$ cannot cover two parts of $\CC$ by minimality of $s$.
  Thus, overall, each part of $\CC'_s$ is covered by the union of at most three consecutive parts of~$\CC$.
  Hence, as we showed that $\bigcup_{T \in \mathcal{S}} C_T$ \crosses at~most~$3k$ consecutive parts of $\CC'_s$, it is contained in $9k$ consecutive parts of $\CC$, say $C_{j}, \ldots, C_{j + 9k - 1}$.
  Thus for any $T \in \mathcal{S}$, $T \setminus \bigcup_{j \leq h \leq j + 9k-1}C_h$ is constant.
  Therefore $R \setminus \bigcup_{j \leq h \leq j + 9k-1}C_h$ contains at most $k$ distinct rows, as $|\mathcal S| \leqslant k$.
\end{proof}

\begin{theorem}\label{no large div imp diagonal}
    For every symmetric $0,1$-matrix $M$ and natural $k$, if $M$ does not have a~$k$-wide division, then $M$ admits a~sequence of~symmetric $2(k+1)$-diagonal divisions.
\end{theorem}
\begin{proof}
  Let $M$ be a~symmetric $n \times n$ matrix that does not admit a~$k$-wide division.
  The finest division $(\R_n, \CC_n)$ of $M$ is 2-diagonal, hence, $2(k+1)$-diagonal.
  Now, we greedily merge some consecutive parts, to form a division sequence in which every division is $2(k+1)$-diagonal.
  Assume, for the sake of contradiction, that after a~partial sequence of symmetric $2(k+1)$-diagonal divisions $(\R_n, \CC_n), \ldots, (\R_s, \CC_s)$ of $M$, no division $(\R, \CC)$ obtained by merging two symmetric pairs of consecutive parts of $(\R_s, \CC_s)$ is $2(k+1)$-diagonal.

Say, $\R_s = (R_1, \ldots, R_s)$ and $\CC_s = (C_1, \ldots, C_s)$.
By assumption, for each $i \in [s-1]$, the division $((R_1, \ldots, R_i \uplus R_{i+1}, \ldots, R_s), (C_1, \ldots, C_i \uplus C_{i+1}, \ldots, C_s))$ is not $2(k+1)$-diagonal.
As $(\R_s, \CC_s)$ is $2(k+1)$-diagonal, $R_i \uplus R_{i+1}$ has to be the (only) $2(k+1)$-wide row part in $(R_1, \ldots, R_i \uplus R_{i+1}, \ldots, R_s)$.
Indeed, for any other row part $R_j$ (with $j \in [i-1] \cup [i+2,s]$), $R_j$ was also a part of $\R_s$, while each part of $(C_1, \ldots, C_i \uplus C_{i+1}, \ldots, C_s)$ contains a~part of~$\CC_s$.
Thus, as $R_j$ is not $2(k+1)$-wide in $(\R_s,\CC_s)$, it is not $2(k+1)$-wide in $((R_1, \ldots, R_i \uplus R_{i+1}, \ldots, R_s), (C_1, \ldots, C_i \uplus C_{i+1}, \ldots, C_s))$.

Hence, in the division $((R_1 \uplus R_2, R_3 \uplus R_4, \ldots, R_{s-1} \uplus R_s), (C_1 \uplus C_2, C_3 \uplus C_4, \ldots, C_{s-1} \uplus C_s))$ (resp. $((R_1 \uplus R_2, R_3 \uplus R_4, \ldots, R_{s-2} \uplus R_{s-1} \uplus R_s), (C_1 \uplus C_2, C_3 \uplus C_4, \ldots, C_{s-2} \uplus C_{s-1} \uplus C_s))$ when $s$ is odd), each row part is at least $k+1$-wide (resp. $k$-wide).
Indeed, $k$ consecutive parts of $(C_1 \uplus C_2, C_3 \uplus C_4, \ldots, C_{s-1} \uplus C_s)$ (resp.~$(C_1 \uplus C_2, C_3 \uplus C_4, \ldots, C_{s-1} \uplus C_{s-1} \uplus C_s)$) are always covered by $2(k+1)$-consecutive parts of $\CC_s$.
Thus $M$ admits a~$k$-wide division, a~contradiction.
\end{proof}

\begin{observation}\label{central removings}
If $M$ is a matrix on which $\mathcal D=(\R = (R_1, \ldots, R_p), \CC = (C_1, \ldots, C_p))$ is a~symmetric $k$-diagonal division, then for every $i \in [p]$, $R_i \setminus \bigcup_{i - k + 1 \leq h \leq i+k -1}C_h$ (for the sake of legibility, an out-of-range value $h$ indexes an empty $C_h$) has less that $k$ distinct rows.
\end{observation}
\begin{proof}
By definition, for each $i \in [p]$, there are $k$ consecutive column parts $C_j, \ldots, C_{j+k-1}$ such that $R_i \setminus \bigcup_{j \leq h \leq j+k-1}C_h$ contains less than $k$ different rows, with $j \leq i \leq j+k-1$.
Hence, $\{C_{i-k+1}, \ldots, C_{i+k-1}\}$ contains $\{C_j, \ldots, C_{j+k-1}\}$, and thus $R_i \setminus \bigcup_{i-k+1 \leq j \leq i+k-1} C_h$ also contains less that $k$ distinct rows.
\end{proof}

\begin{theorem}\label{building the contraction}
    If a symmetric $0,1$-matrix $M$ admits a~sequence of symmetric $k$-diagonal divisions, then $\stw(M) \leq 4k^3$.
\end{theorem}
\begin{proof}
Let $M$ be an $n \times n$ $0,1$-matrix, and $(\R_n, \CC_n), \ldots, (\R_1, \CC_1)$, a sequence of symmetric $k$-diagonal divisions.
For any $s \in [2,n]$, let $\R_s = (R_1, \ldots, R_s)$ and  $\CC_s = (C_1, \ldots, C_s)$.
By \Cref{central removings}, for any $i \in [s]$, $R_i \setminus \bigcup_{i-k+1 \leq h \leq i+k-1} C_h$ contains less that $k$ distinct rows.
Let $(\R'_s, \CC'_s)$ be the symmetric partition of $M$ such that for each $P \in \R'_s$, there is $R_i \in \R_s$ with $P \subseteq R_i$, and $P$ is a maximal subset of equal rows of $R_i \setminus \bigcup_{i-k+1 \leq h \leq i+k-1} C_h$.

As each $R \in \R_s$ is split into at most $k$ parts in $\R'_s$, the stretch value of $(\R'_s, \CC'_s)$ is at most $(2k-1)k$.
By construction, $(\R'_ {s-1}, \CC'_{s-1})$ is a coarsening of $(\R'_s, \CC'_s)$.
Indeed, every $R'_i \in \R'_s$ is an equivalence class of $R_i \setminus \bigcup_{i-k+1 \leq h \leq i+k-1} C_h$.
Thus it is contained in an equivalence class of the row part of $(\R_{s-1},\CC_{s-1})$ containing $R'_i$.

To go from $(\R'_s, \CC'_s)$ to $(\R'_{s-1}, \CC'_{s-1})$, we perform any symmetric partition sequence.
As any part in the latter partition contains at most $2k$ parts in the former one (all the possible parts partitioning two row parts $R_i$ and $R_{i+1}$ that are merged), the stretch value in between these two partitions is bounded by $2k^2 \cdot 2k$.
Hence $\stw(M) \leq 4k^3$.
\end{proof}

\begin{theorem}\label{bdd-wideness-impl-bdd-stw}
If a~symmetric $0,1$-matrix $M$ does not admit a~$k$-wide division, then $\stw(M) \leqslant 32(k+1)^3$.
\end{theorem}
\begin{proof}
  Indeed, by \Cref{no large div imp diagonal}, $M$ admits a~sequence of $2(k+1)$-diagonal divisions.
  Applying \Cref{building the contraction} on this sequence yields a~witness of stretch-width $4 \cdot (2(k+1))^3 = 32 (k+1)^3$.
\end{proof}

\begin{reptheorem}{thm:approx-stw}
  Given an integer $k$ and a symmetric $n \times n$ $0,1$-matrix $M$, there is an $n^{O(1)}$-time algorithm that outputs a~symmetric partition sequence witnessing that $\stw(M) = O(k^3)$ or correctly reports that $\stw(M) > k$.
\end{reptheorem}
\begin{proof}
Given any $n \times n$ $0,1$-matrix $M$, division $(\R = (R_1, \ldots, R_p), \CC = (C_1, \ldots, C_p))$ of $M$, part $R \in \R$, and integer $q$, one can decide in polynomial time if the row part $R_i$ is $q$-wide in $(\R, \CC)$.
Indeed, it suffices to check for every $i-q+1 \leq j \leq i$, if $R_i \setminus \bigcup_{j \leq h \leq j+q-1} C_h$ contains at least $q$ different rows, which can be done in $n^{O(1)}$.
Thus, in time $n^{O(1)}$, one can check if the division $(\R, \CC)$ is $q$-diagonal, and for every $i \in [p-1]$ if there is a $q$-diagonal division of the form $((R_1, \dots, R_i \uplus R_{i+1}, \dots, R_p), (C_1, \dots, C_i \uplus C_{i+1}, \dots, C_p))$.

Let $M$ an $n \times n$ symmetric $0, 1$-matrix, and $k$ an integer.
We start from the finest division $(\R_n, \CC_n)$.
If at some point we have a division $(\R_s = (R_1, \ldots, R_s), \CC_s = (C_1, \ldots, C_s))$ such that none of the divisions $(R_1, \dots, R_i \uplus R_{i+1}, \dots, R_s), (C_1, \dots, C_i \uplus C_{i+1}, \dots, C_s)$ are $2(9k+1)$-diagonal, then by~\Cref{no large div imp diagonal}, the division $((R_1 \uplus R_2, \ldots, R_{s-1} \uplus R_s)$, $(C_1 \uplus C_2, \ldots, C_{s-1} \uplus C_s))$ is $9k$-wide.
Thus by \Cref{no large div imp diagonal}, $\stw(M) > k$.

Otherwise, in time $n^{O(1)}$, we get~a sequence of $2(9k+1)$-diagonal divisions $(\R_n, \CC_n), \ldots,$  $(R_1, \CC_1)$.
At this point, we can find a sequence witnessing that $\stw(M) = O(k^3)$, by the proof of~\Cref{building the contraction}.
We build the symmetric partitions $(\R'_n, \CC'_n), \ldots, (\R'_1,\CC'_1)$ of $M$, where for each $P \in \R'_s$, there is $R_i \in \R_s$ with $P \subseteq R_i$, and $P$ is a maximal subset of equal rows of $R_i \setminus \bigcup_{i-k+1 \leq h \leq i+k-1} C_h$.
%Note that $(\R_s, \CC_s)_{1 \leqslant s \leqslant n-1}$ is not a proper sequence for stretch-width: Two consecutive partitions are not changing by a single merge of parts.
The partition $\R'_s$ (and its symmetric $\CC'_s$) can be found in polynomial time: For each $R_i \in \R_s$, one can sort the rows of $R_i \setminus \bigcup_{i-k+1 \leq h \leq i+k-1} C_h$ by lexicographic order, and obtain the desired equivalence classes of equal rows.

\Cref{building the contraction} ensures that any sequence going from $(\R'_s, \CC'_s)$ to $(\R'_{s-1}, \CC'_{s-1})$ maintains a~stretch value of $O(k^3)$. 
\end{proof}

%\section{Overlap graph, subdivisions, and small separators}\label{sec:stw-bd-deg}

%The main purpose of this section is to understand the properties of stretch-width on sparse graphs.
%They exists many different notions of sparsity, but we will only focus on graphs of \emph{bounded degree}.

\section{Overlap graph}\label{sec:overlap}

Consider an ordered graph $(G, \prec)$, and think of $\prec$ as a~\emph{left-to-right} order (with the smallest vertex being the leftmost one).
For any edge $e \in E(G)$, we denote by $L(e)$ (resp.~$R(e)$) the left (resp.~right) endpoint of~$e$.
Given two edges $e, f \in E(G)$, we say that $e$ is \emph{left of}~$f$ if $L(e) \preceq L(f)$, and $e$ is~\emph{strictly left of}~$f$ if $L(e) \prec L(f)$.
By extension, we say that $X \subset E(G)$ is \emph{left of} (resp.~\emph{strictly left of}) $Y \subset E(G)$ if for every $e \in X$ and $f \in Y$, $L(e) \preceq L(f)$ (resp.~$L(e) \prec L(f)$).
If $u, v$ are vertices of $(G,\prec)$, we denote by $[u,v]$ the set of vertices that are, in $\prec$, at least $u$ and at most $v$.
We also denote by $[\leftarrow,u]$ (resp.~$[u,\rightarrow]$) the set of vertices that are at most $u$ (resp.~at least $u$). 

We say that two edges $e, f$ are \emph{crossing} if $L(e) \prec L(f) \prec R(e) \prec R(f)$ (or symmetrically) and we denote $e \times f$ this relation.
Observe that two edges sharing an endpoint are not crossing.
The relation $\times$ is symmetric and anti-reflexive, hence defines an~undirected graph on $E(G)$.
We denote by $\Ov(G,\prec)$ the graph $(E(G), \times)$.
$\Ov(G,\prec)$ is called the \emph{overlap graph} of $(G, \prec)$; see~\cref{fig:overlap-graph}.

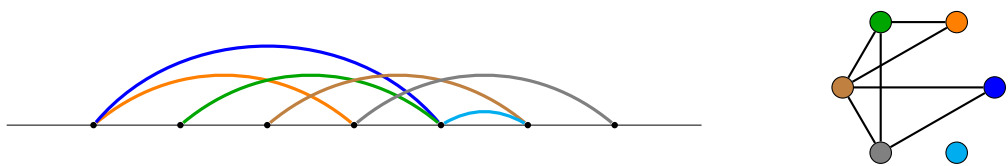
\begin{figure}[h!]
  \centering
  \begin{tikzpicture}[vertex/.style={fill,circle,inner sep=0.03cm},lvertex/.style={draw,circle,inner sep=0.1cm}]
    \def\n{7} % number of vertices
    \def\l{8} % length of the representation

    % vertices of (G,<)
    \draw[very thin] (0,0) -- (\l+\l/\n,0) ;
    \foreach \i in {1,...,\n}{
      \node[vertex] (\i) at (\l * \i/\n,0) {} ;
    }
    % edges of (G,<)
    \def\g{green!65!black}
    \foreach \i/\j/\b/\c in {1/4/40/orange,1/5/50/blue,2/5/40/\g,3/6/40/brown,4/7/40/gray,5/6/30/cyan}{
      \draw[very thick,\c] (\i) to [bend left=\b] (\j) ;
    }

    \def\r{1}
    \pgfmathsetmacro\z{360/6} % divide by the number of edges of G
    \begin{scope}[xshift=12cm,yshift=0.5cm]    
      % vertices of Ov(G,<)
      \foreach \i/\c/\p in {0/blue/bl,1/orange/or,2/\g/g,3/brown/br,4/gray/gr,5/cyan/cy}{
        \node[lvertex,fill=\c] (\p) at (\i * \z:\r) {} ;
      }
      % edges of Ov(G,<)
      \foreach \i/\j in {bl/br,bl/gr,or/g,or/br,g/br,g/gr,br/gr}{
        \draw[thick] (\i) -- (\j) ;
      }
  \end{scope}
\end{tikzpicture}
  \caption{An ordered graph (left) and its overlap graph (right).}
  \label{fig:overlap-graph}
\end{figure}

We relate the structure of $\Ov_\prec(G)$ and the stretch-width of $G$ among bounded-degree graphs, by proving the following theorem:
\begin{theorem}
A~class $\mathcal C$ of ordered graphs of bounded degree has bounded stretch-width if and only if $\{\Ov(G,\prec)~|~G \in \mathcal C\}$ does not admit $K_{t,t}$ subgraph, for some integer $t$.
\end{theorem}

The next two lemmas prove the forward implication, by considering a special point in the partition sequence.
The last lemma of this section proves the backward implication, using the matrix characterization of~\cref{sec:matrices}.
We say that a~$K_{t,t}$ subgraph of $\Ov(G,\prec)$ is \emph{clean} if the sides of the $K_{t,t}$ are $X, Y \subset E(G)$ such that $X$ is strictly left of $Y$.

\begin{lemma}\label{from Ktt to simple Ktt}
  For every ordered graph $(G, \prec)$, if $\Ov(G,\prec)$ contains a~$K_{t, t}$ as a subgraph, then~$\Ov(G,\prec)$ contains a clean~$K_{\lfloor t/2 \rfloor, \lfloor t/2 \rfloor}$ subgraph.
\end{lemma}
\begin{proof}
  Assuming that $\Ov(G,\prec)$ has a~$K_{t,t}$ subgraph, there is two disjoint sets $X, Y \subset E(G)$ each of size $t$ such that for every $x \in X$ and $y \in Y$, $x \times y$.
Let $L(x_1) \preceq L(x_2) \preceq \ldots \preceq L(x_t)$ be the elements of $X$, and $L(y_1) \preceq L(y_2) \preceq \ldots \preceq L(y_t)$, the elements of $Y$.
As $x_{\lfloor t/2 \rfloor}$ and $y_{\lfloor t/2 \rfloor}$ are crossing, either $L(x_{\lfloor t/2 \rfloor}) \prec L(y_{\lfloor t/2 \rfloor})$ or $L(y_{\lfloor t/2 \rfloor}) \prec L(x_{\lfloor t/2 \rfloor})$.
The sides of the clean $K_{\lfloor t/2 \rfloor, \lfloor t/2 \rfloor}$ are $\{x_1, \ldots, x_{\lfloor t/2 \rfloor}\}$ and $\{y_{\lfloor t/2 \rfloor}, \ldots, y_t\}$ in the former case, and $\{y_1, \ldots, y_{\lfloor t/2 \rfloor}\}$ and $\{x_{\lfloor t/2 \rfloor}, \ldots, x_t\}$ in the latter.  
\end{proof}

\begin{lemma}\label{bdd deg-bdd stw impl bdd Ktt}
For any ordered graph $(G, \prec)$, if $\Delta(G) \leqslant d$ and $\stw(G, \prec) \leqslant t$, then $\Ov(G,\prec)$ does not contain $K_{N, N}$ with $N = 4td^2$ as a subgraph.
\end{lemma}
\begin{proof}
We will prove the contrapositive.
Let $(G, \prec)$ be an ordered graph of maximum degree at most~$d$.
We suppose that $\Ov(G,\prec)$ contains a $K_{N, N}$ as a subgraph, with $N=4td^2$.
By \Cref{from Ktt to simple Ktt}, there are two sets $X, Y \in E(G)$ forming a~clean $K_{N/2, N/2}$ of $\Ov(G,\prec)$.

Let $v_1$ (resp.~$v_2$) be the rightmost vertex among left endpoints of edges in $X$ (resp.~$Y$), and let $v_3$ be the rightmost vertex among (right) endpoints of edges in~$X$; see~\cref{fig:X-Y-vi}.
Observe that for any edge $e \in X$, $L(e) \preceq v_1$, and $v_2 \prec R(e) \preceq v_3$.
The relation $v_2 \prec R(e)$ holds because every edge of $X$ crosses every edge of $Y$.
In addition, for any edge $f \in Y$, $v_1 \prec L(f) \preceq v_2$, and $v_3 \prec R(f)$.
%As inequalities are large, we can suppose that each $c_i$ is incident to an edge of $Y$.

\begin{figure}[h!]
  \centering
  \begin{tikzpicture}[vertex/.style={fill,circle,inner sep=0.02cm}]
    \def\s{30}
    \def\l{12}
    \draw[very thin] (0,0) -- (\l+\l/\s,0) ;
    \foreach \i in {1,...,\s}{
      \node[vertex] (\i) at (\l * \i/\s,0) {} ;
    }
    %edges
    \def\g{green!65!black}
    \foreach \i/\j/\b/\c in {1/17/28/blue,1/19/30/blue,2/17/25/blue,2/21/30/blue,3/21/25/blue,4/16/28/blue,5/22/20/blue,7/16/25/blue,
    8/24/30/\g,9/24/25/\g,9/29/30/\g,10/27/25/\g,10/30/30/\g,12/30/30/\g,14/24/20/\g,14/28/25/\g}{
      \draw[thick,\c] (\i) to [bend left=\b] (\j) ;
    }
    %v1,v2,v3
    \foreach \i/\p in {7/$v_1$,14/$v_2$,22/$v_3$}{
      \node at (\l * \i / \s,-0.25) {\p} ;
    }
    %X,Y
    \foreach \i/\p in {2/$X$,10.5/$Y$}{
      \node at (\i,1) {\p} ;
    }
  \end{tikzpicture}
  \caption{The edge subsets $X$ (blue) and $Y$ (green) forming a~clean biclique $K_{N/2, N/2}$ of $\Ov(G,\prec)$ (here, with $N=16$), and the vertices $v_1, v_2, v_3 \in V(G)$.
  The vertices non-incident to a~blue or green edge are in $V(G) \setminus V(H)$.}
  \label{fig:X-Y-vi}
\end{figure}
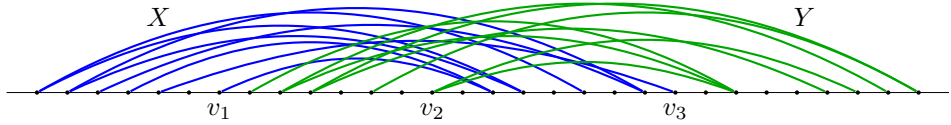

We consider $H$, the subgraph of $(G,\prec)$ induced by the endpoints of edges in $X \cup Y$, and let $h$ be $|V(H)|$.
We will show that the stretch-width of $H$ (hence that of $G$) is at least~$t$.
In the remaining of the proof, the intervals of vertices and the lengths of edges are all with respect to~$H$.
The intervals $[\leftarrow, v_1[, [v_1, v_2[, ]v_2, v_3[$ and $[v_3, \rightarrow]$ are all of size at least $N/(2d)$.
Indeed, observe that every vertex of $H$ has at most $d$ incident edges.       
Thus, the $N/2$ left (resp.~right) endpoints of edges in $X$ (resp.~$Y$) make for at least $\frac{1}{d} 	 N/2$ distinct vertices. 

Therefore, every edge in $X \cup Y$ has length at least $N/(2d)$.
Let $\P_h, \ldots, \P_1$ be any partition sequence of $H$, and let $i \in [h]$ be the maximum integer such that $d+1$ vertices are contained in a~single part $P \in \P_i$.
As $\Delta(H) \leqslant \Delta(G) \leqslant d$, every part $Q \in \P_i$ adjacent to $P$ is such that $P$ and $Q$ are inhomogeneous.
Let $e=uv$ be an edge of $X \cup Y$ with at least one endpoint in $P$, say $u$, and let $P'$ (possibly equal to~$P$) be the part of $\P_i$ containing the other endpoint,~$v$.
The span of $N_{\R(\P_i)}(P)$ contains the interval $I=[u,v]$ (or $I=[v,u]$ if $v \prec u$).
As $I$ has length at least $N/(2d)$ and every part of $\P_i$ has size at most~$2d$, the number of parts of $\P_i$ \crossing $I$ (hence, in particular with $N_{\R(\P_i)}(P)$) is at least $N/(4d^2)$.
Thus, $\stw(G,\prec) \geqslant \stw(H,\prec) \geqslant \frac{N}{4d^2}=t$.
\end{proof}

\begin{replemma}{lem:weakly-sparse-overlap}\label{bdd overlap Ktt impl bdd stw}
For every ordered graph $(G, \prec)$ and positive integer $N$, if $\Ov(G,\prec)$ does not contain $K_{N, N}$ as a subgraph, then $\stw(G, \prec) \leqslant 32 (2N+1)^3$.
\end{replemma}
\begin{proof}
Let $(G, \prec)$ be an ordered graph such that $\Ov(G,\prec)$ does not contain $K_{N, N}$ as a~subgraph, and let $M$ be the adjacency matrix of $(G, \prec)$.
We prove that $\stw(M) \leqslant 32 (2N+1)^3$.

Suppose, for the sake of contradiction, that $\stw(M) > 32 (2N+1)^3$.
By \Cref{bdd-wideness-impl-bdd-stw}, there is a~$2N$-wide division $(\R= \{R_1, \ldots, R_k\}, \CC= \{C_1, \ldots, C_k\})$ of $M$.
In particular, for any row $R_i$, $R_i \setminus C_{i-N+1}, \ldots, C_{i+N-1}$ contains more that $2N$ different rows.
Let $D$ be the union of the zones $R_i \cap C_j$ such that~$|i - j| < N$, that is, the $2N-1$ ``longest'' diagonals of zones of the division $(\mathcal{R}, \mathcal{C})$.
As, for every $i \in [k]$, the number of distinct rows in $R_i \setminus D$ (resp. distinct columns in $C_i \setminus D$) is at least $2N$, $R_i \setminus D$ (resp. $C_i \setminus D$) contains at least $2N$ 1-entries.

To simplify the coming notations, let denote by $\lVert M' \rVert$ the number of 1-entries of any submatrix $M'$ of $M$.
For example, $\lVert R_i \setminus D \rVert \geqslant 2N$.
Observe that $R_i$ (resp. $C_j$) is split by $D$ in at most two sets $R_i^{\leftarrow}$ and $R_i^{\rightarrow}$ (resp. $C_j^{\uparrow}$ and $C_j^{\downarrow}$), namely, $R_i^{\leftarrow} = \bigcup_{j \leqslant i - N} R_i \cap C_j$ and $R_i^{\leftarrow} = \bigcup_{j \geqslant i + N} R_i \cap C_j$; see \Cref{matrix drawing}.

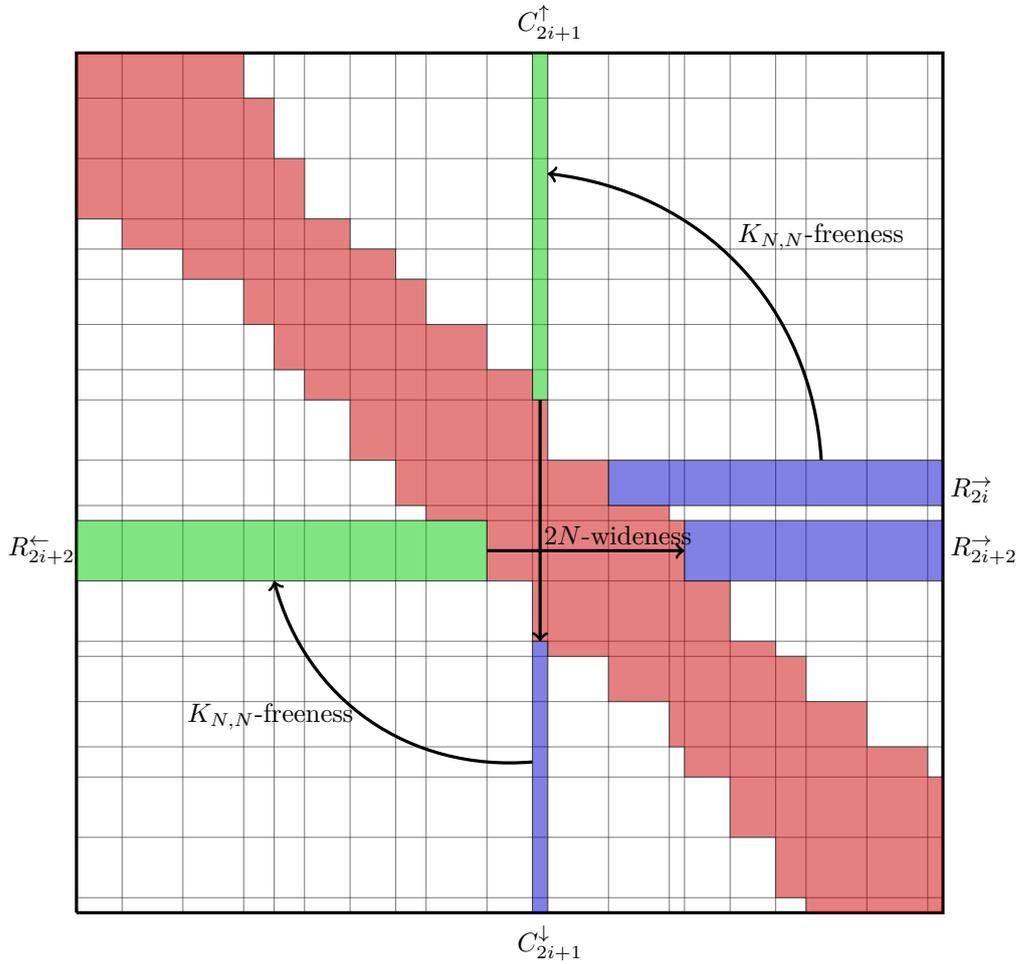
\begin{figure}[h!]
\centering
\begin{tikzpicture}[scale = 0.2]

\draw[very thick] (0, 0)--(57,0)--(57,57)--(0,57)--(0,0);
\foreach \i in {3, 7, 11, 13, 15, 18, 21, 23, 27, 30, 31, 35, 39, 40, 43, 46, 48, 52, 56, 57, 57}{
\draw[opacity = 0.5] (\i, 57) -- (\i, 0);
\draw[opacity = 0.5] (0, 57 - \i) -- (57, 57 - \i);
}
\filldraw[fill=red!80!black, opacity = 0.5]
(0, 57) \foreach \i/\j in {11/57,  11/54,  13/54,  13/50,  15/50,  15/46,  18/46,  18/44,  21/44,  21/42,  23/42,  23/39,  27/39,  27/36,  30/36,  30/34,  31/34,  31/30,  35/30,  35/27,  39/27,  39/26,  40/26,  40/22,  43/22,  43/18,  46/18,  46/17,  48/17,  48/14,  52/14,  52/11,  56/11,  56/9,  57/9,  57/5,  57/5,  57/0,  52/0,  52/0, 48/0,  48/1, 46/1,  46/5, 43/5,  43/9, 40/9,  40/11, 39/11,  39/14, 35/14,  35/17, 31/17,  31/18, 30/18,  30/22, 27/22,  27/26, 23/26,  23/27, 21/27,  21/30, 18/30,  18/34, 15/34,  15/36, 13/36,  13/39, 11/39,  11/42, 7/42,  7/44, 3/44,  3/46, 0/46,  0/50}{-- (\i, \j)} --cycle;

\filldraw[fill=blue!80!black, opacity = 0.5]
(35, 30) -- (35, 27) -- (57, 27) -- (57, 30) -- cycle;
\node[text width=3cm] at (65,28) {$R_{2i}^\rightarrow$};

\filldraw[fill=green!80!black, opacity = 0.5]
(30, 34) -- (31, 34) -- (31, 57) -- (30, 57) -- cycle;
\node[text width=3cm] at (36.5,59) {$C_{2i+1}^\uparrow$};

\filldraw[fill=blue!80!black, opacity = 0.5]
(30, 18) -- (31, 18) -- (31, 0) -- (30, 0) -- cycle;
\node[text width=3cm] at (36.5,-2) {$C_{2i+1}^\downarrow$};

\filldraw[fill=green!80!black, opacity = 0.5]
(27, 26) -- (27, 22) -- (0, 22) -- (0, 26) -- cycle;
\node[text width=3cm] at (3,24) {$R_{2i+2}^\leftarrow$};

\filldraw[fill=blue!80!black, opacity = 0.5]
(40, 26) -- (40, 22) -- (57, 22) -- (57, 26) -- cycle;
\node[text width=3cm] at (65,24) {$R_{2i+2}^\rightarrow$};

\draw[->, very thick] (49, 30) to [bend right=40] (31, 49);
\draw[->, very thick] (30, 10) to [bend right=-40] (13, 22);
\node[text width=3cm] at (51,44.8) {$K_{N, N}$-freeness};
\node[text width=3cm] at (14.8,13) {$K_{N, N}$-freeness};

\draw[->, very thick] (30.5, 34) to (30.5, 18);
\draw[->, very thick] (27, 24) to (40, 24);
\node[text width=3cm] at (38.3,25) {$2N$-wideness};
\end{tikzpicture}
\caption{Visual depiction of the proof of \cref{lem:weakly-sparse-overlap}. The red zones represent $D$, the blue zones contain more than $N$ 1-entries, and the green zones have fewer that $N$ 1-entries. The invariant leading to a contradiction (that $R_{2i}^\rightarrow$ has more than $N$ 1-entries) propagates by invoking twice $K_{N,N}$-freeness followed by $2N$-wideness.} 
\label{matrix drawing}
\end{figure}

Observe that for every $i, j$ such that $i + 1 \leqslant j < i + N$, each 1-entry of $R_i^{\rightarrow}$ (resp. $C_i^{\downarrow}$) and 1-entry of $C_j^{\uparrow}$ (resp. $R_j^{\rightarrow}$) correspond to crossing edges in $(G,\prec)$.
As $\Ov(G,\prec)$ does not contain any $K_{N, N}$ subgraph we have, for every $i, j$ such that $i + 1 \leqslant j < i + N$:
\begin{compactenum}
  \item \label{it:1} $\min(\lVert R_i^{\rightarrow} \rVert, \lVert C_j^{\uparrow} \rVert) < N$, and
  \item \label{it:2} $\min(\lVert C_i^{\downarrow} \rVert, \lVert R_j^{\leftarrow} \rVert) < N$.
\end{compactenum}
Indeed, if the first item does not hold, $N$ 1-entries in $R_i^{\rightarrow}$ and $N$ 1-entries in $C_j^{\uparrow}$ form the two sides of a~$K_{N,N}$.

We finally prove by induction on $i$ that, while $2i \leqslant k$, the property \emph{$\lVert R_{2i}^{\rightarrow} \rVert > N$}, henceforth called $(\mathcal Q_i)$, holds.
%(note that this is absurd because $R_j^{\rightarrow}$ ends up being empty).
Note that $R_{0}^{\leftarrow}$ is empty.
Thus $\lVert R_{0}^{\rightarrow} \rVert \geqslant 2N > N$, hence $(\mathcal Q_0)$ holds.
Now assume that $(\mathcal Q_i)$ holds.
By the first item, we have $\lVert C_{2i + 1}^{\uparrow}\rVert < N$.
Thus $\lVert C_{2i + 1}^{\downarrow} \rVert > N$, since $$C_{2i + 1} \setminus D = C_{2i + 1}^{\downarrow} \cup C_{2i + 1}^{\uparrow}~\text{and}~\lVert C_{2i + 1} \setminus D \rVert \geqslant 2N.$$
Symmetrically, by the second item, $\lVert R_{2i + 2}^{\leftarrow} \rVert < N$, and hence $\lVert R_{2i + 2}^{\rightarrow} \rVert > N$.
Thus $(\mathcal Q_{i+1})$ holds.
As $ R_{k - N +1}^{\rightarrow}$ is empty, $(\mathcal Q_i)$ can no longer be true when $2i \geqslant k - N + 1$, a~contradiction.
Therefore $\stw(M) \leqslant 32 (2N)^3$. 
\end{proof}

\section{Subdivisions}\label{sec:subd}

When subdividing the edges of an~ordered graph, there is a~simple way of updating its vertex ordering without creating larger bicliques in its overlap graph.

\begin{lemma}\label{lem:single-subd}
  Let $(G, \prec)$ be an ordered graph, and $H$ be obtained by subdividing an edge of~$G$.
  There is an order $\prec'$ such that, for every integer $t$, if $\Ov(G,\prec)$ has no $K_{t,t}$ subgraph, then $\Ov(H,\prec')$ has no $K_{t,t}$ subgraph.
\end{lemma}

\begin{proof}
  Let $e=uv$ be the edge of $G$ subdivided to form $H$, and let $w \in V(H)$ be the new vertex resulting from this subdivision.
  The total order $\prec'$ is obtained from $\prec$, by adding $w$ next to~$u$, say, just to its right.
  This way $\Ov(H,\prec')$ is simply $\Ov(G,\prec)$ plus an isolated vertex.
  Indeed the edge $uw \in E(H)$ is an isolated vertex in $\Ov(H,\prec')$, since $u$ and $w$ are consecutive along $\prec'$, whereas $wv \in E(H)$ crosses the same edges as $uv$ was crossing. 
\end{proof}

We now define a long subdivision process that is actually ``erasing'' large bicliques in the overlap graph of a~bounded-degree graph.
Let $uv$ be an edge of an ordered graph $(G,\prec)$, with $h$ vertices between $u$ and $v$, say, $u \prec u_1 \prec u_2 \prec \ldots \prec u_h \prec v$.
We describe an $h+1$-subdivision of $uv$ in $(G,\prec)$ that we call \emph{flattening of $uv$}.
We delete $uv$, and create $h+1$ new vertices $w_1, \ldots, w_{h+1}$ such that $u \prec w_1 \prec u_1 \prec w_2 \prec u_2 \prec \ldots \prec w_h \prec u_h \prec w_{h+1} \prec v$.
We then create the edges $uw_1$, $w_iw_{i+1}$ for every $i \in [h]$, and $w_{h+1}v$.
We may say that these \emph{edges stem} from $uv$.
An \itsubd of $(G,\prec)$ chooses a~total order on the edges of~$G$, and iteratively flattens the edges of $G$ in this order (note that the created edges are \emph{not} flattened themselves); see~\cref{subdivisions example}.

\begin{figure}[h!]
\centering
\begin{tikzpicture}[vertex/.style={fill,circle,inner sep=0.02cm}]
  \begin{scope}
    \def\s{5}
    \def\l{4}
    \draw[very thin] (\l / \s-0.2,0) -- (\l+0.2,0) ;
    \foreach \i in {1,...,\s}{
      \node[vertex] (\i) at (\l * \i/\s,0) {} ;
      \coordinate (c\i) at (\l * \i/\s + 0.5 * \l/\s,0) {} ;
    }
    %edges
    \def\g{green!65!black}
    \foreach \i/\j/\b/\c in {1/3/50/blue,1/5/50/\g,2/4/50/orange}{
      \draw[thick,\c] (\i) to [bend left=\b] (\j) ;
    }
    \draw[thin,densely dashed,blue] (1) to [bend left=50] (c1) to [bend left=50] (c2) to [bend left=50] (3) ;
  \end{scope}

  \begin{scope}[xshift=4.25cm]
    \def\s{7}
    \def\l{4}
    \draw[very thin] (\l / \s-0.2,0) -- (\l+0.2,0) ;
    \foreach \i in {1,...,\s}{
      \node[vertex] (\i) at (\l * \i/\s,0) {} ;
      \coordinate (c\i) at (\l * \i/\s + 0.5 * \l/\s,0) {} ;
    }
    %edges
    \def\g{green!65!black}
    \foreach \i/\j/\b/\c in {1/2/50/blue,2/4/50/blue,4/5/50/blue,1/7/50/\g,3/6/50/orange}{
      \draw[thick,\c] (\i) to [bend left=\b] (\j) ;
    }
    \draw[very thin,densely dashed,orange] (3) to [bend left=50] (c3) to [bend left=50] (c4) to [bend left=50] (c5) to [bend left=50] (6) ;
  \end{scope}

  \begin{scope}[xshift=8.5cm]
    \def\s{10}
    \def\l{4}
    \draw[very thin] (\l / \s-0.2,0) -- (\l+0.2,0) ;
    \foreach \i in {1,...,\s}{
      \node[vertex] (\i) at (\l * \i/\s,0) {} ;
      \coordinate (c\i) at (\l * \i/\s + 0.5 * \l/\s,0) {} ;
    }
    %edges
    \def\g{green!65!black}
    \foreach \i/\j/\b/\c in {1/2/50/blue,2/5/50/blue,5/7/50/blue,1/10/50/\g,3/4/50/orange,4/6/50/orange,6/8/50/orange,8/9/50/orange}{
      \draw[thick,\c] (\i) to [bend left=\b] (\j) ;
    }
    \draw[very thin,densely dashed,\g] (1) to [bend left=50] (c1) to [bend left=50] (c2) to [bend left=50] (c3) to [bend left=50] (c4) to [bend left=50] (c5) to [bend left=50] (c6) to [bend left=50] (c7) to [bend left=50] (c8) to [bend left=50] (c9) to [bend left=50] (10) ;
  \end{scope}

  \begin{scope}[yshift=-1cm,xshift=3.5cm]
    \def\s{19}
    \def\l{6}
    \draw[very thin] (\l / \s-0.2,0) -- (\l+0.2,0) ;
    \foreach \i in {1,...,\s}{
      \node[vertex] (\i) at (\l * \i/\s,0) {} ;
      \coordinate (c\i) at (\l * \i/\s + 0.5 * \l/\s,0) {} ;
    }
    %edges
    \def\g{green!65!black}
    \foreach \i/\j/\b/\c in {1/3/50/blue,3/9/50/blue,9/13/50/blue,
      1/2/50/\g,2/4/50/\g,4/6/50/\g,6/8/50/\g,8/10/50/\g,10/12/50/\g,12/14/50/\g,14/16/50/\g,16/18/50/\g,18/19/50/\g,
      5/7/50/orange,7/11/50/orange,11/15/50/orange,15/17/50/orange}{
      \draw[thick,\c] (\i) to [bend left=\b] (\j) ;
    }
  \end{scope}
\end{tikzpicture}
     \caption{An \itsubd. Created edges have the color of the edge they stem from.}
    \label{subdivisions example}
\end{figure}
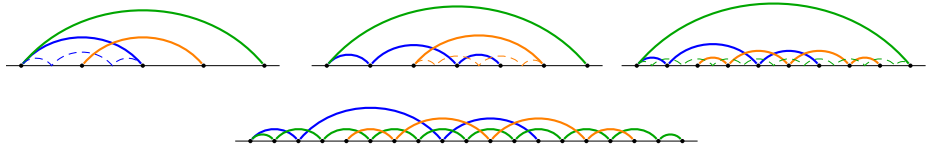

\begin{lemma}\label{lem:decr-biclique}
  Any \itsubd $(G',\prec')$ of an~ordered graph $(G,\prec)$ of maximum degree~$d$, is such that $\Ov(G',\prec')$ has no $K_{2d+2,2d+2}$ subgraph.
\end{lemma}
\begin{proof}
  Assume for the sake of contradiction that $\Ov(G',\prec')$ has a~$K_{2d+2,2d+2}$ subgraph.
  Then by~\cref{from Ktt to simple Ktt}, $\Ov(G',\prec')$ has a~clean $K_{d+1,d+1}$ subgraph.
  Let $X, Y$ be the two sides of this clean biclique, where $X$ is left of $Y$.
  As every vertex of $G'$ (like $G$) is incident to at~most $d$~edges, there is $\{x_1,x_2\} \subseteq X$ and $\{y_1,y_2\} \subseteq Y$ such that $L(x_1) \preceq L(x_2) \prec L(y_1) \prec L(y_2) \prec R(x_i) \prec R(x_{3-i}) \prec R(y_j) \preceq R(y_{3-j})$ with $i, j \in [2]$.

  As $x_1$ and $x_2$ cross $y_1$ and $y_2$, there is no $i, j \in [2]$ such that $x_i$ and $y_j$ stem from the same edge of $G$.
  We can thus assume without loss of generality that the last edge among $x_1, x_2, y_1, y_2$ to be created is in $X$ (since the argument is symmetric if this happens in~$Y$), i.e., $x_i$ for some $i \in [2]$.
  When $x_i$ is created, the vertices $L(y_1)$ and $L(y_2)$ already exist and form a~non-trivial interval since $L(y_1) \prec L(y_2)$.
  This contradicts the construction of the \itsubd, since $x_i$ jumps over $[L(y_1), L(y_2)]$, when it should have at least created an intermediate vertex in $[L(y_1), L(y_2)]$.
\end{proof}

We are now equipped to show the main result of this section, that exponentially-long subdivisions of bounded-degree graphs have bounded stretch-width. 

\begin{reptheorem}{thm:subd}\label{thm:long-subd-bd-deg}
  Every $(\geqslant n 2^m)$-subdivision of every $n$-vertex $m$-edge graph $G$ of maximum degree~$d$ has stretch-width at~most~$32(4d+5)^3$.
\end{reptheorem}
\begin{proof}
  Let $G$ be any graph of $\mathcal C$ with $n$~vertices and $m$~edges, and let $G''$ be any $(\geqslant n 2^m)$-subdivision of $G$.
  Choose an arbitrary order $\prec$ of $V(G)$.
  Let $(G',\prec')$ be the \itsubd of $(G,\prec)$, performed with an arbitrary order on the edges $G$.
  By~\cref{lem:decr-biclique}, $\Ov(G',\prec')$ has no $K_{2d+2,2d+2}$ subgraph.
  Every edge of $G$ is subdivided at most $n 2^m$ times by the process of \itsubd.
  By~\cref{lem:single-subd}, the edges of $G'$ can be further subdivided to obtain $G''$ such that $\Ov(G'',\prec'')$ has no $K_{2d+2,2d+2}$ subgraph, for some vertex ordering $\prec''$.
  Therefore, by \cref{lem:weakly-sparse-overlap}, $\stw(G'',\prec'') \leqslant 32(4d+5)^3$, and in particular, $\stw(G'') \leqslant 32(4d+5)^3$.
\end{proof}

\begin{corollary}\label{cor:stw-dg-vs-tw}
There are graph classes with bounded stretch-width and maximum degree, and yet unbounded treewidth.
\end{corollary}
\begin{proof}
  Consider the family $\Gamma_1, \Gamma_2, \ldots$, where $\Gamma_k$ is the $k^2 2^{2k(k-1)}$-subdivision of the $k \times k$-grid, for every positive integer $k$.
  The graphs from this family have degree at most~4, and stretch-width at~most 296352, but unbounded treewidth since $\tw(\Gamma_k)=k$. 
\end{proof}

The above argument gives an example of $n$-vertex graphs with bounded degree and stretch-width, and treewidth $\Omega(\sqrt{\log n})$.
We can do better by picking the vertex ordering $\prec$, and the order on the edges (for the \itsubd) more carefully.
\begin{figure}[h!]
\centering
\begin{tikzpicture}[vertex/.style={fill,circle,inner sep=0.02cm}]
  \begin{scope}
    \def\s{16}
    \def\l{10}
    \draw[very thin] (0,0) -- (\l+\l/\s,0) ;
    \foreach \i in {1,...,\s}{
      \node[vertex] (\i) at (\l * \i/\s,0) {} ;
      \coordinate (c\i) at (\l * \i/\s + 0.5 * \l/\s,0) {} ;
    }
    %edges
    \foreach \i/\j in {1/2,2/3,3/4, 5/6,6/7,7/8, 9/10,10/11,11/12, 13/14,14/15,15/16}{
      \draw[thick,blue] (\i) to [bend left=50] (\j) ;
    }
    \foreach \i/\j in {1/5,5/9,9/13, 2/6,6/10,10/14, 3/7,7/11,11/15, 4/8,8/12,12/16}{
      \draw[thick,green!60!black] (\i) to [bend left=50] (\j) ;
    }
  \end{scope}
\end{tikzpicture}
     \caption{The $4 \times 4$ grid ordered row by row, with the horizontal edges in blue, and vertical edges in green.}
    \label{subd:grid}
\end{figure}
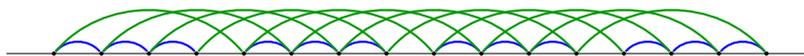
We simply order the grid row by row, and from left to right within each row; see~\cref{subd:grid}.
We perform the \itsubd of this ordered grid, with the following edge ordering.
First we flatten every horizontal edge (in blue), in any order.
When this is done, the total number of vertices has less than doubled.
Then we flatten every vertical edge (in green) from left to right.
It can be observed that, starting from the $k \times k$ grid, we now obtain an \itsubd with less than $2^{ck}$ vertices, for some constant $c$.
Thus, there are $n$-vertex graphs with bounded degree and stretch-width, and treewidth $\Omega(\log n)$.

\section{Classes with bounded $\Delta+\stw$ have logarithmic treewidth}\label{sec:tw-bound}

For any edge $e$ of an ordered graph $(G, \prec)$, we denote $e^i = \ ]L(e), R(e)[$, the \emph{interior of $e$}, and $e^o := [\leftarrow, L(e)[ \cup ]R(e), \rightarrow]$, the \emph{exterior of $e$}; note that $L(e)$ and $R(e)$ are neither part of $e^i$ nor of $e^o$.
The \emph{length} of $e$ according to $\prec$ is
$\ell (e, \prec) = R(e) - L(e)$. When $F$ is a set of edges we define $\ell (F, \prec)$ to be the maximum length of an edge of $F$.

We say that a set $C$ of vertices is a \emph{$c$-balanced separator of $G$} when there is a~$c$-balanced separation $(A, B)$ of $G$ such that $C = A \cap B$.
In an ordered graph $(G, \prec)$, a~set of vertices $C$ is called~\emph{left/right $c$-balanced separator} if there is a~$c$-balanced separation $(A, B)$, with $C = A \cap B$, $A$ contains the leftmost $c \cdot n$ vertices, and $B$ contains the rightmost $c \cdot n$ vertices.
Given a subset $U$ of vertices of $(G, \prec)$, we define the ordered induced subgraph $(G, \prec)[U]$ as the subgraph obtained from $G$ by removing the vertices not in $U$, and restraining $\prec$ to $U$. When $\prec$ is clear from the context, we use the notation $G[U]$.

We start with an observation on paths going from the interior to the exterior of an edge.
\begin{observation}\label{edge separation}
For every ordered graph $G$ and for every edge $e \in E(G)$, any path $P$ starting in $e^i$ and ending in $e^o$ contains an endpoint of $e$, or an edge crossing~$e$.
\end{observation}
\begin{proof}
  Let $e=uv$ be an edge of $G$, and $P = v_1, \dots v_k$, a path with $v_1 \in e^i$ and $v_k \in e^o$.
  Let~$j$ be the smallest index such that $v_j \notin e^i$.
  As $V(G) = \{u,v\} \cup e^i \cup  e^o$, we either have $v_j \in \{u,v\}$, or $v_j \in e^o$ and $v_{j-1} \in e^i$, hence $v_{j-1}v_j$ is an edge of $P$ crossing~$e$.
\end{proof}

%From this simple fact, we can deduce powerful properties.
To simplify the notations, if the vertices of $(G, \prec)$ are $v_1 \prec \dots \prec v_n$, we will write $G\langle i, j \rangle$ instead of $G[[v_i, v_j]]$ and $\langle i, j \rangle$ instead of $[v_i, v_j]$. 

We say that a set $S$ of edges of $(G, \prec)$ is a \emph{rainbow} if for every pair $e, f$ of $S$, $e^i \subset f^i$ or $f^i \subset e^i$.
Notice that a rainbow induces an independent set in $\Ov(G, \prec)$.
When $S$ contains $t$ edges we say that $S$ is a \emph{$t$-rainbow}, or a rainbow of order $t$.
The following is an application of Dilworth's theorem on permutation graphs (or the Erd\H{o}s-Szekeres theorem); see for instance \cite[Lemma 2.1]{Cerny-07}.

\begin{lemma}[\cite{Cerny-07}]\label{poly Ramsey}
Let $(G, \prec)$ be an ordered graph, such that $\Ov(G,\prec)$ does not contain a~clique on $t$ vertices.
Then, for every vertex $v$ of $V(G)$, for every set $F$ of edges from $[\leftarrow, v[$ to $]v, \rightarrow]$ we have $|F| \leqslant kt$ where $k$ is the maximum order of a rainbow in $F$.
\end{lemma}

For any rainbow $S$, we set $S^i:= \bigcup_{e \in S} e^i$, that is, $S^i$ is the interior of the edge of $S$ with maximal interior.
If $S$ is the empty set, by convention, $S^i$ is also empty.
A \emph{rainbow over $v$} is a rainbow $S$ contained in the set of edges from $[\leftarrow, v[$ to $]v, \rightarrow]$.
A \emph{maximum rainbow over $v$} is a rainbow of maximum cardinality among the rainbows over $v$.

In the next lemma, we find, by induction on the minimum length of a~maximum rainbow over some vertex, a small (but not always balanced) left/right separator.

\begin{lemma}\label{separation main lemma}
Let $(G, \prec)$ be an ordered graph such that $\Ov(G,\prec)$ does not contain a~$K_{t, t}$ subgraph.
Then, if $S$ is a maximum rainbow over $v \in V(G)$, there is a vertex $x \in S^i \cup \{v\}$ and a~set~$U$ that separates $[\leftarrow, x[$ from $]x, \rightarrow]$ with $|U| = O(t^2  \log \ell(S, \prec))$.
\end{lemma}
\begin{proof}
  First, note that we write \emph{$S^i \cup \{v\}$} instead of simply \emph{$S^i$} in the event that $S=\emptyset$ (in which case $v$ is not contained in $S^i$).
  We define $h(l) = 2^{\lceil \log (l + 1) \rceil + 1}$ and $g(t) = 6t^2 + 3t$.
  %One can check that $h(0) = 2$ and $h(l) \geqslant 2 h(\lfloor l/2 \rfloor)$. 
  Thus $g(t) = O(t^2)$ and $h(l) = O(l)$.

We prove by induction on $l$ that for every ordered graph $(G, \prec)$ for which $\Ov(G,\prec)$ has no $K_{t, t}$ subgraph, and for every vertex $v \in V(G)$, if $S$ is a~maximum rainbow over~$v$ such that $\ell(S, \prec) \leqslant l$, then there exists a vertex $x \in S^i \cup \{v\}$ and a set $U$ that separates $[\leftarrow, x[$ from $]x, \rightarrow]$ with $|U| \leqslant g(t) \log h(l)$.

If $l = 0$, there is no edge over~$v$.
Hence, $S$ is empty and $\{v\}$ separates $[\leftarrow, v[$ from $]v, \rightarrow]$. 
Assume that the property holds for any $k < l$.
We show the property also holds for~$l$.
%Let $(G, \prec)$ be an ordered graph without $K_{t, t}$ in $\Ov(G, \prec)$. 
Let~$v$ be a vertex of $G$, and $S$ be a maximum rainbow over $v$ such that $\ell(S, \prec) \leq l$.
Let $F$ be the set of edges from $[\leftarrow, v[$ to $]v, \rightarrow]$.
We make the following case distinction: the rainbow $S$ contains either at most $3t$ edges or more than $3t$ edges.

In the former case, \Cref{poly Ramsey} ensures that $|F| \leqslant 2t |S| \leqslant 6t^2$.
Indeed, $S$ is a rainbow of $F$ of maximum cardinality, and $\Ov(G, \prec)$ does not contain any $K_{t, t}$ subgraph, hence any $K_{2t}$.
Thus, consider $L_F := \{L(f) \ | \ f \in  F \} \cup \{v \}$.
The set $L_F$ separates $[\leftarrow, v[$ from $]v, \rightarrow]$ since any edge with exactly one endpoint in $[\leftarrow, v[$ has an endpoint in $L_F$.
As $|L_F| \leqslant 6 t^2 +1 \leqslant g(t) \leqslant g(t)\log h(l) $, the property holds.

We now handle the case when $|S|>3t$.
Let $\{e_1, \ldots, e_k\}$ be the edges in $S$ with $e^i_{j+1} \subsetneq e^i_{j}$ for every $j \in [k-1]$ (we are sorting the edges from longest to smallest).
Let $X$ the set of edges from $[L(e_1), L(e_{3t})]$ to $[R(e_1), R(e_{3t})]$.
Set $X$ cannot contain a rainbow of size $3t + 1$, otherwise $S$ would not be maximum.
As $X$ does not contain a $K_{2t}$, \Cref{poly Ramsey} implies that $|X| \leqslant 6t^2$.
Let $Y_1$ be the set of edges from $e_{1}^o$ to $e_{t}^i$,
$Y_2$ be the set of edges from $e_{t}^o$ to $e_{2t}^i$, and
$Y_3$ be the set of edges from $e_{2t}^o$ to $e_{3t}^i$; see~\cref{separation lemma help}.

\begin{figure}[h!]
\begin{tikzpicture}[scale = 0.375]
\tikzset{
    dot/.style={circle,inner sep=1pt, minimum size = 1pt, fill,label={[xshift=0.cm, yshift = -0.6cm] \footnotesize{#1}},name=#1},
}

\filldraw[opacity = 0.7, fill=red!80!black]
(7.5, 0) -- (0, 0) to[bend left = 40] (35, 0) -- (24, 0) to[bend left = -40] (7.5,0);

\foreach \i/\j in {0.0/35.0,  0.5/34.5,  0.5/34.0,  0.5/33.5,  1.0/33.0,  1.5/32.5,  1.5/32.0,  1.5/31.5,  1.5/31.0,  2.0/30.5,  2.0/30.5,  2.5/30.0,  3.0/29.5,  3.5/29.5,  4.0/29.5,  4.0/29.5,  4.5/29.0,  4.5/29.0,  4.5/28.5,  5.0/28.0,  5.0/27.5,  5.5/27.0,  5.5/27.0,  5.5/26.5,  6.0/26.0,  6.0/25.5,  6.5/25.0,  7.0/25.0,  7.0/24.5,  7.5/24.0} {
	\draw (\i, 0) to [bend left = 40] (\j, 0);
}

\filldraw[opacity = 0.5, fill=green!80!black]
(7.5, 0) -- (24, 0) to[bend left = 40] (28, 0) -- (35, 0) to[bend left = -40] (7.5,0);
\filldraw[opacity = 0.5, fill=green!80!black]
(7.5, 0) -- (24, 0) to[bend left = -40] (0, 0) -- (5, 0) to[bend left = 40] (7.5,0);

\foreach \i/\j in {0.0/35.0,  2.0/30.5, 5.0/28.0,  7.5/24.0} {
	\draw[very thick, blue] (\i, 0) to [bend left = 40] (\j, 0);
}
\node[dot] at (18,0) {} ;
\node at (18,-0.65) {$v$} ;
\node[dot= $L(e_1)$] (0) at (0,0) {} ;
\node[dot= $L(e_t)$] (0) at (2.0,0) {} ;
\node[dot= $L(e_{2t})$] (0) at (5,0) {} ;
\node[dot= $L(e_{3t})$] (0) at (7.5,0) {} ;
\node[dot= $R(e_1)$] (0) at (35,0) {} ;
\node[dot] (0) at (30.5,0) {} ;
\node[text width=3cm] at (34.5, -0.6) {\footnotesize{$R(e_t)$}};
\node[dot= $R(e_{2t})$] (0) at (28,0) {} ;
\node[dot= $R(e_{3t})$] (0) at (24,0) {} ;
\node[dot] at (3.5,0) {} ;
\node at (3.5,-0.65) {$x$};
\node[dot] (0) at (30,0) {} ;
\node at (30,-0.65) {$y$} ;

\draw[very thick, red] (18, 6.5) -- (25, 7.5);
\node[text width=3cm] at (29, 7.5) {$X$};

\draw[very thick, green] (18, 1) -- (15, -0.5);
\node[text width=3cm] at (17.8, -0.7) {$Y_3$};
\end{tikzpicture}
\caption{The black (or blue) edges are $e_1, e_2, \ldots, e_{3t} \in S$, and the edges $e_1, e_t, e_{2t}, e_{3t}$ are in thick blue.
  The red arch represents the edge set $X$, while the green shape symbolizes the edge set~$Y_3$.
  We will now remove the left endpoints of these edges sets (as well as those of $Y_1$ and $Y_2$), define the vertices $x \in [L(e_t),L(e_{2t})]$ and $y \in [R(e_{2t}),R(e_t)]$, and apply the induction on the \emph{shorter} maximum rainbow over $x$ or over $y$.
}
\label{separation lemma help}
\end{figure}
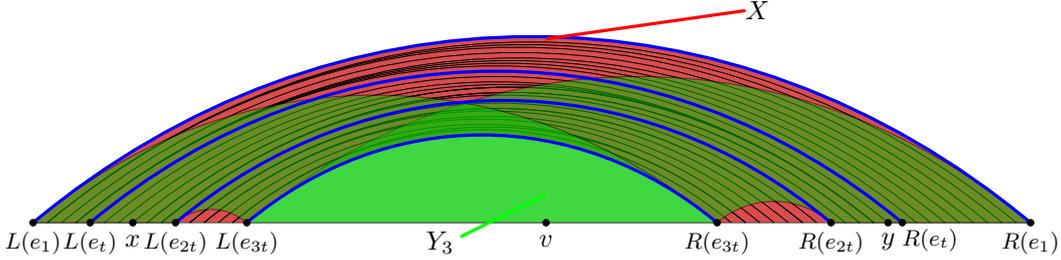

We have $|Y_1| \leqslant t$, as each edge of $Y_1$ crosses $e_{1}, \ldots, e_{t}$.
Similarly, $|Y_2| \leqslant t$ and $|Y_3| \leqslant t$.
Let $L_{X \cup Y} := \{L(f) \ | \ f \in X \cup Y_1 \cup Y_2 \cup Y_3\}$.
Observe that $|L_{X \cup Y}| \leqslant 6 t^2 + 3t = g(t)$.

Let $G' = G - L_{X \cup Y}$.
In $G'$, for any $w \in ]L(e_{t}), L(e_{2t})]$, the edges over $w$ have both endpoints in $[L(e_1), L(e_{3t})]$.
Indeed let $f$ be an edge over $w$ and assume that $f$ has an endpoint outside of $[L(e_1), L(e_{3t})]$.
By assumption, this endpoint is in $e_1^o \cup e_{3t}^i \cup [R(e_{3t}), R(e_1)]$ (the set of remaining vertices).
As $f$ is over $w$, the left endpoint of $f$ is either in $[\leftarrow, e_1]$ or in $[e_1, w[$, and its right endpoint is in $[w, \rightarrow ]$.
Suppose the left endpoint of $f$ is in $[\leftarrow, e_1]$.
Then $f$ is in $Y_1$ when its right endpoint is in $e_t^i$, and in $X$, otherwise.
Assume now that the left endpoint of $f$ is in $[\leftarrow, e_1]$, then its right endpoint is in $[L(e_{3t}), \rightarrow]$ (since by assumption, $f$ is not in $[L(e_1, L(e_{3t})]$), and $f$ is an edge of $Y_3$ when its right endpoint is in $e_{3t}^i$, and an edge of $X$ when its right endpoint is in $[R(e_{3t}), \rightarrow]$.

In any case, we reach the contradiction that $f$ is in $Y_1 \cup Y_3 \cup X$.
Note that if there is no vertex of $G'$ in $]L(e_{t}), L(e_{2t})]$, $L_{X \cup Y}$ separates $G'$: As we removed $X, Y_1, Y_2$ and $Y_3$, there is no edges between $[\leftarrow, L(e_{t})]$ and $]L(e_{2t}), \rightarrow]$.
The same holds symmetrically when there is no vertex of $G'$ in~$[R(e_{2t}), R(e_{t})[$.

Thus we consider two vertices $x$ and $y$ respectively in (the non-empty sets) $]L(e_{t}), L(e_{2t})] \cap V(G')$ and $[R(e_{2t}), R(e_{t})[ \cap V(G')$.
In $G'$, edges over $x$ are contained in $[L(e_1), L(e_{3t})]$, and edges over $y$, in $[R(e_{3t}), R(e_1)]$.
Hence no edge is both over $x$ and $y$.       
Consider $S_x$ and $S_y$ two maximum rainbows over $x$ and $y$, respectively.
We have $\ell(S_x, \prec) + \ell(S_y, \prec) \leqslant \ell(S, \prec)$.
Assume, without loss of generality, that $\ell(S_x, \prec) \leqslant \ell(S_y, \prec)$.
We have $\ell(S_x, \prec) \leqslant \lfloor l/2 \rfloor$.
By the induction hypothesis, there is a vertex $x'$ in $S_x^i \cup \{x\}$ and a set $U$ that separates $[\leftarrow, x'[$ from $]x', \rightarrow]$ in $G'$ with $|U| \leqslant g(t) \log h(\lfloor l/2 \rfloor)$.

Thus $U \cup L_{X\cup Y}$ separates $[\leftarrow, x'[$ from $]x', \rightarrow]$ in $G$, and $|U \cup L_{X\cup Y}| \leqslant g(t)\log h(\lfloor l/2 \rfloor)+g(t)=g(t)(\lceil \log (\lfloor l/2 \rfloor+1) \rceil + 2)\leq g(t)(\lceil \log (l+1) \rceil + 1)= g(t) \log h(l)$.
\end{proof}

Leveraging \cref{separation main lemma}, we now show how to obtain a~relatively small \emph{balanced} separator.

\begin{theorem}\label{separation Thm}
For any ordered graph $(G, \prec)$, if $\Ov(G,\prec)$ does not contain any $K_{t, t}$ subgraph, then $G$ contains a 1/12-balanced separator of order at most $\gamma t^2 \log n$, for some constant $\gamma$.
\end{theorem}
\begin{proof}
We first examine the case when there is a vertex $v \in V(G)$ and a maximum rainbow $S$ over $x$ such that at least $3t$ edges of $S$ have length in $[n/12, 11n/12]$.
We denote by $f_1, \ldots, f_{3t}$ the $3t$ longest edges of $S$ of length at most $11n/12$, still ordered from longest ($f_1$) to shortest ($f_{3t}$).
As in the proof of \Cref{separation main lemma}, we consider $X$ the edges from $[L(f_1), L(f_{3t})]$ to $[R(f_{3t}), R(f_1)]$, and \Cref{poly Ramsey} ensures that $|X| \leqslant 6t^2$.
Let $Y_1$ be the set of edges from $f_{1}^o$ to $f_{t}^i$, $Y_2$ be the set of edges from $f_{t}^o$ to $f_{2t}^i$, and $Y_3$ be the set of edges from $f_{2t}^o$ to $f_{3t}^i$.
We again have that $Y_1, Y_2$ and $Y_3$ are of size at~most~$t$, otherwise $\Ov(G, \prec)$ would contain a $K_{t, t}$ subgraph.
Let $L_{X \cup Y} = \{L(z) \ | \ z \in X \cup Y_1 \cup Y_2 \cup Y_3 \}$.

Let $G' = G - L_{X \cup Y}$.
As observed in~\Cref{separation main lemma}, the separator is only simpler when $]L(f_t), L(f_{2t})] \cap V(G')$ or $[R(f_{2t}), R(f_{t})[ \cap V(G')$ are empty.
Hence we only deal with the case when both sets are non-empty (in the other cases, a~subset of our eventual separator works).
Thus, let $x \in ]L(f_t), L(f_{2t})] \cap V(G')$ and $y \in [R(f_{2t}), R(f_{t})[ \cap V(G')$.
    
By construction of $G'$, any edge over $x$ is contained in $[L(f_1), L(f_{3t})]$, or is going from $[\leftarrow, L(f_1)[$ to $]R(f_1), \rightarrow]$.
Thus, by applying \Cref{separation main lemma} on $G_x = G'[\leftarrow, R(f_1)]$ and $x$, we find a set $U_x$ separating $[\leftarrow, u_x]$ from $[u_x, \rightarrow]$, where $u_x$ is a vertex in $[L(f_1), L(f_{3t})]$.
Similarly, considering $G_y = G[L(f_1), \rightarrow]$ and $y$, we find a set $U_y$ separating $[\leftarrow, u_y]$ from $[u_y, \rightarrow]$, with $u_y \in [R(f_{3t}), R(f_1)]$.
The set $U_x \cup U_y$ separates $]u_x, u_y[$ from $[\leftarrow, u_x[ \cup [u_y, \rightarrow]$ in $G'$.
As the length between $u_x$ and $u_y$ is between $n/12$ and $11n/12$, $L_{X \cup Y} \cup U_x \cup U_y$ is a $1/12$-balanced separator of $G$ of size $6t^2 + 3t + 2 \cdot O(t^2 \log n) = O(t^2 \log n)$.

Now, we deal with the case when for any vertex $v$ of $G$, the number of edges of length in $[n/12, 11n/12]$ in a (maximum) rainbow over $v$ is less than~$3t$.
Let $v_1 \prec v_2 \prec \ldots \prec v_n$ be the vertices of $G$, and say, $v=v_i$.
In particular, the set $M_v$ of edges over $v$ going from $\langle n/12, i - n/12 \rangle$ to $\langle i+n/12, 11n/12 \rangle$ is of size at most $6t^2$ by \Cref{poly Ramsey}.
Set $x = v_{\lfloor n/3 \rfloor}$ and $y = v_{\lfloor2n/3 \rfloor}$.
Let $A = \{L(e) \ | \ e \in M_x \cup M_y\}$ and let $H = G - A$.
Consider $H_x = H\langle 1, 11n/12 \rangle$, and $H_y = H \langle n/12, n \rangle$.
Then the length of a maximum rainbow over $x$ (resp.~$y$) in $H_x$ (resp.~$H_y$) is at~most $n/12$.
Indeed, any edge over $x$ (resp. $y$) of length more than $11n/12$ is not contained in $H_x$ (resp. $H_y$), and any edge over~$x$ (resp.~$y$) of length in $[n/12,11n/12]$ has been deleted when removing $A$.

Hence we can apply \Cref{separation main lemma} on $H_x$ and $x$, and on $H_y$ over $y$.
This yields two sets $U_x, U_y$ of size $O(t^2 \log n)$ and two vertices $w_x, w_y$ such that the indices of $w_x$ and $x$ (resp. $w_y$ and $y$) are at distance at most $n/12$, and such that $U_x$ separates $[\leftarrow, w_x[$ from $]w_x, \rightarrow]$ in $H_x$, while $U_y$ separates $[\leftarrow, w_y[$ from $]w_y, \rightarrow]$ in $H_y$.
The set $A \cup U_x \cup U_y$ then separates $]w_x, w_y[$ from $[\leftarrow,w_u[ \cup ]w_y,\rightarrow]$ in $G$, and $|A \cup U_x \cup U_y| = O(t^2 \log n)$.
Notice that $]w_x, w_y[$ contains at least $\lfloor n/3 \rfloor- 2 n/12 \geq \lfloor n/6 \rfloor$ vertices because there is $ \lfloor n/3 \rfloor$ vertices between $x$ and $y$, and $w_x$ is separated from $x$ by at most $n/12$ vertices (and symmetrically for $y$).
In addition, there is at most $n/3 + 2n/12 = n/2$ vertices in $]w_x, w_y[$.
Thus $A \cup U_x \cup U_y$ is a $1/6$-balanced separator.
\end{proof}

Pipelining~\cref{bdd deg-bdd stw impl bdd Ktt} and the previous theorem, we get small balanced separators for graphs of bounded degree and stretch-width. 

\begin{theorem}\label{thm:sn-bound}
Let $G$ be any graph on $n$ vertices, such that $\Delta(G) \leqslant d$ and $\stw(G) \leqslant t$.
Then $G$ contains a $1/12$-balanced separator of size at most $\gamma (4td^2)^2 \log n$, for a constant $\gamma$.
\end{theorem}
\begin{proof}
  Let $G$ a graph on $n$ vertices, such that $\stw(G) \leqslant t$ and $\Delta(G) \leqslant d$.
  Let $\prec$ be an order such that $(G, \prec)$ has stretch-width $t$.
  By \cref{bdd deg-bdd stw impl bdd Ktt}, $\Ov(G, \prec)$ does not admit any $K_{4td^2, 4td^2}$ as a subgraph.
  Hence~\cref{separation Thm} ensures the existence of a $1/12$-balanced separator of $(G, \prec)$, hence of $G$, of size $\gamma (4td^2)^2 \log n$.
\end{proof}

By~\Cref{lem:tw-sn} (and the remark following it) and \cref{thm:sn-bound}, we obtain the bound on the treewidth of a graph of bounded degree and bounded stretch-width.

\begin{reptheorem}{thm:tw-bound}
There is a~$c$ such that for every graph $G$, $\tw(G) \leqslant c \Delta(G)^4 \stw(G)^2 \log |V(G)|$.
\end{reptheorem}

We draw two algorithmic consequences from \cref{thm:tw-bound}.
First, Existential Counting Modal Logic (see the introduction for a~description) can be model-checked in polynomial time in classes where both the maximum degree and the stretch-width is bounded.

\begin{repcorollary}{cor:ecml}
  Problems definable in ECML (resp. ECML+C) can be solved in polynomial time (resp. randomized polynomial time) in bounded-degree graphs of bounded stretch-width.
\end{repcorollary}
\begin{proof}
  Let $\mathcal C$ be a class with bounded $\Delta+\stw$.
  By~\cref{thm:tw-bound}, there is a~constant~$c$ such that every graph $G \in \mathcal C$ has treewidth at~most $c \log n$.
  There is a~single-exponential 2-approximation of treewidth~\cite{Korhonen21}, which returns here a~tree-decomposition of width at~most $2c \log n$ in time $2^{O(c \log n)}=n^{O(1)}$.
  Any problem definable in ECML (resp.~ECML+C) can be solved in time $2^{O(w)}$ (resp. randomized time $2^{O(w)}$, with a Monte Carlo algorithm) where $w$ is the width of a~tree-decomposition of the input graph~\cite{Pilipczuk11}.
  This gives an algorithm in (randomized) time $2^{O(2c \log n)}=n^{O(1)}$ for any problem definable in ECML(+C) on~$\mathcal C$.
\end{proof}

The second consequence is a~subexponential-time algorithm for \smis (and any problem with a~similar branching rule on high-degree vertices) in graphs of bounded stretch-width.
\begin{repproposition}{prop:subexp}
  There is an algorithm that solves \mis in graphs of bounded stretch-width with running time $2^{\Tilde{O}(n^{4/5})}$. 
\end{repproposition}
\begin{proof}
  Let $\mathcal C$ be a class of bounded stretch-width.
  By~\cref{thm:tw-bound}, for every $n$-vertex graph $G \in \mathcal C$, the treewidth of $G$ is in $O(\Delta(G)^4 \log n)$.
  We use a~standard branching rule that Turing-reduces the problem to \smis on subinstances of small maximum degree, hence small treewidth. 
  
  While there is in the current graph a~vertex $v$ of degree at~least $n^{1/5}$, we branch on two options: either (first branch) we put $v$ in (an initially empty set) $I$, and remove its closed neighborhood, or (second branch) we remove $v$ from the current graph (without adding it to~$I$).
  With the former choice, the number of vertices drops by at least $n^{1/5}$, and it drops by~1 in the latter.
  The former outcome can only happen at most $n^{4/5}$ times, as we started with $n$ vertices.
  Hence the branching tree has at most ${n \choose n^{4/5}}=2^{\Tilde{O}(n^{4/5})}$ leaves.

  The treewidth of the graph at every leaf is in $O(n^{4/5}\log n)$ since there are no more vertices of degree at~least $n^{4/5}$.
  As explained in the proof of~\cref{cor:ecml}, we can solve such instances in time $2^{O(n^{4/5}\log n)}$, and append the corresponding $I$ to the output.
  The overall running time is $2^{\Tilde{O}(n^{4/5})} \cdot 2^{O(n^{4/5}\log n)} = 2^{\Tilde{O}(n^{4/5})}$.
\end{proof}

%\bibliography{biblio}

\begin{thebibliography}{10}

\bibitem{Abrishami22}
Tara Abrishami, Maria Chudnovsky, Sepehr Hajebi, and Sophie Spirkl.
\newblock Induced subgraphs and tree-decompositions {III}.
  {T}hree-path-configurations and logarithmic tree-width.
\newblock {\em Advances in Combinatorics}, 2022.

\bibitem{BarilSlides}
Ambroise Baril, Miguel Couceiro, and Victor Lagerkvist.
\newblock Linear bounds between cliquewidth and component twin-width and
  applications, 2023.
\newblock URL: \url{https://ramics20.lis-lab.fr/slides/slidesAmbroise.pdf}.

\bibitem{Berge21}
Pierre Berg{\'{e}}, {\'{E}}douard Bonnet, and Hugues D{\'{e}}pr{\'{e}}s.
\newblock Deciding twin-width at most 4 is {NP}-complete.
\newblock In Mikolaj Bojanczyk, Emanuela Merelli, and David~P. Woodruff,
  editors, {\em 49th International Colloquium on Automata, Languages, and
  Programming, {ICALP} 2022, July 4-8, 2022, Paris, France}, volume 229 of {\em
  LIPIcs}, pages 18:1--18:20. Schloss Dagstuhl - Leibniz-Zentrum f{\"{u}}r
  Informatik, 2022.
\newblock \href {https://doi.org/10.4230/LIPIcs.ICALP.2022.18}
  {\path{doi:10.4230/LIPIcs.ICALP.2022.18}}.

\bibitem{Berge23}
Pierre Berg{\'{e}}, {\'{E}}douard Bonnet, Hugues D{\'{e}}pr{\'{e}}s, and
  R{\'{e}}mi Watrigant.
\newblock Approximating highly inapproximable problems on graphs of bounded
  twin-width.
\newblock In Petra Berenbrink, Patricia Bouyer, Anuj Dawar, and
  Mamadou~Moustapha Kant{\'{e}}, editors, {\em 40th International Symposium on
  Theoretical Aspects of Computer Science, {STACS} 2023, March 7-9, 2023,
  Hamburg, Germany}, volume 254 of {\em LIPIcs}, pages 10:1--10:15. Schloss
  Dagstuhl - Leibniz-Zentrum f{\"{u}}r Informatik, 2023.
\newblock \href {https://doi.org/10.4230/LIPIcs.STACS.2023.10}
  {\path{doi:10.4230/LIPIcs.STACS.2023.10}}.

\bibitem{Bergougnoux23}
Benjamin Bergougnoux, Jan Dreier, and Lars Jaffke.
\newblock A logic-based algorithmic meta-theorem for mim-width.
\newblock In Nikhil Bansal and Viswanath Nagarajan, editors, {\em Proceedings
  of the 2023 {ACM-SIAM} Symposium on Discrete Algorithms, {SODA} 2023,
  Florence, Italy, January 22-25, 2023}, pages 3282--3304. {SIAM}, 2023.
\newblock \href {https://doi.org/10.1137/1.9781611977554.ch125}
  {\path{doi:10.1137/1.9781611977554.ch125}}.

\bibitem{Bonamy23}
Marthe Bonamy, Edouard Bonnet, Hugues D{\'{e}}pr{\'{e}}s, Louis Esperet, Colin
  Geniet, Claire Hilaire, St{\'{e}}phan Thomass{\'{e}}, and Alexandra Wesolek.
\newblock Sparse graphs with bounded induced cycle packing number have
  logarithmic treewidth.
\newblock In Nikhil Bansal and Viswanath Nagarajan, editors, {\em Proceedings
  of the 2023 {ACM-SIAM} Symposium on Discrete Algorithms, {SODA} 2023,
  Florence, Italy, January 22-25, 2023}, pages 3006--3028. {SIAM}, 2023.
\newblock \href {https://doi.org/10.1137/1.9781611977554.ch116}
  {\path{doi:10.1137/1.9781611977554.ch116}}.

\bibitem{twin-width3}
{\'{E}}douard Bonnet, Colin Geniet, Eun~Jung Kim, St{\'{e}}phan Thomass{\'{e}},
  and R{\'{e}}mi Watrigant.
\newblock Twin-width {III:} {Max Independent Set, Min Dominating Set, and
  Coloring}.
\newblock In Nikhil Bansal, Emanuela Merelli, and James Worrell, editors, {\em
  48th International Colloquium on Automata, Languages, and Programming,
  {ICALP} 2021, July 12-16, 2021, Glasgow, Scotland (Virtual Conference)},
  volume 198 of {\em LIPIcs}, pages 35:1--35:20. Schloss Dagstuhl -
  Leibniz-Zentrum f{\"{u}}r Informatik, 2021.
\newblock \href {https://doi.org/10.4230/LIPIcs.ICALP.2021.35}
  {\path{doi:10.4230/LIPIcs.ICALP.2021.35}}.

\bibitem{twin-width2}
{\'{E}}douard Bonnet, Colin Geniet, Eun~Jung Kim, St{\'{e}}phan Thomass{\'{e}},
  and R{\'{e}}mi Watrigant.
\newblock {Twin-width II: small classes}.
\newblock {\em Combinatorial Theory}, 2(2), 2022.
\newblock URL: \url{https://escholarship.org/uc/item/9cs265b9}, \href
  {https://doi.org/http://dx.doi.org/10.5070/C62257876}
  {\path{doi:http://dx.doi.org/10.5070/C62257876}}.

\bibitem{twin-width4}
{\'{E}}douard Bonnet, Ugo Giocanti, Patrice~Ossona de~Mendez, Pierre Simon,
  St{\'{e}}phan Thomass{\'{e}}, and Szymon Torunczyk.
\newblock Twin-width {IV:} ordered graphs and matrices.
\newblock In Stefano Leonardi and Anupam Gupta, editors, {\em {STOC} '22: 54th
  Annual {ACM} {SIGACT} Symposium on Theory of Computing, Rome, Italy, June 20
  - 24, 2022}, pages 924--937. {ACM}, 2022.
\newblock \href {https://doi.org/10.1145/3519935.3520037}
  {\path{doi:10.1145/3519935.3520037}}.

\bibitem{twin-width6}
{\'E}douard Bonnet, Eun~Jung Kim, Amadeus Reinald, and St{\'e}phan
  Thomass{\'e}.
\newblock Twin-width {VI:} the lens of contraction sequences.
\newblock In {\em Proceedings of the 2022 Annual ACM-SIAM Symposium on Discrete
  Algorithms (SODA)}, pages 1036--1056. SIAM, 2022.

\bibitem{twin-width1}
{\'{E}}douard Bonnet, Eun~Jung Kim, St{\'{e}}phan Thomass{\'{e}}, and
  R{\'{e}}mi Watrigant.
\newblock Twin-width {I:} tractable {FO} model checking.
\newblock {\em J. {ACM}}, 69(1):3:1--3:46, 2022.
\newblock \href {https://doi.org/10.1145/3486655} {\path{doi:10.1145/3486655}}.

\bibitem{reduced-bdw}
{\'{E}}douard Bonnet, O{-}joung Kwon, and David~R. Wood.
\newblock Reduced bandwidth: a qualitative strengthening of twin-width in
  minor-closed classes (and beyond).
\newblock {\em CoRR}, abs/2202.11858, 2022.
\newblock URL: \url{https://arxiv.org/abs/2202.11858}, \href
  {http://arxiv.org/abs/2202.11858} {\path{arXiv:2202.11858}}.

\bibitem{Courcelle90}
Bruno Courcelle.
\newblock {The monadic second-order logic of graphs. I. Recognizable sets of
  finite graphs}.
\newblock {\em Information and Computation}, 85(1):12 -- 75, 1990.
\newblock URL:
  \url{http://www.sciencedirect.com/science/article/pii/089054019090043H},
  \href {https://doi.org/https://doi.org/10.1016/0890-5401(90)90043-H}
  {\path{doi:https://doi.org/10.1016/0890-5401(90)90043-H}}.

\bibitem{Courcelle00}
Bruno Courcelle, Johann~A. Makowsky, and Udi Rotics.
\newblock Linear time solvable optimization problems on graphs of bounded
  clique-width.
\newblock {\em Theory Comput. Syst.}, 33(2):125--150, 2000.
\newblock \href {https://doi.org/10.1007/s002249910009}
  {\path{doi:10.1007/s002249910009}}.

\bibitem{Dvorak19}
Zdenek Dvor{\'{a}}k and Sergey Norin.
\newblock Treewidth of graphs with balanced separations.
\newblock {\em J. Comb. Theory, Ser. {B}}, 137:137--144, 2019.
\newblock \href {https://doi.org/10.1016/j.jctb.2018.12.007}
  {\path{doi:10.1016/j.jctb.2018.12.007}}.

\bibitem{GJ79}
Michael~R. Garey and David~S. Johnson.
\newblock {\em Computers and Intractability: {A} Guide to the Theory of
  NP-Completeness}.
\newblock W. H. Freeman, 1979.

\bibitem{Golovach23}
Petr~A. Golovach, Giannos Stamoulis, and Dimitrios~M. Thilikos.
\newblock Model-checking for first-order logic with disjoint paths predicates
  in proper minor-closed graph classes.
\newblock In Nikhil Bansal and Viswanath Nagarajan, editors, {\em Proceedings
  of the 2023 {ACM-SIAM} Symposium on Discrete Algorithms, {SODA} 2023,
  Florence, Italy, January 22-25, 2023}, pages 3684--3699. {SIAM}, 2023.
\newblock \href {https://doi.org/10.1137/1.9781611977554.ch141}
  {\path{doi:10.1137/1.9781611977554.ch141}}.

\bibitem{Gurski00}
Frank Gurski and Egon Wanke.
\newblock The tree-width of clique-width bounded graphs without
  \emph{K\({}_{\mbox{n, n}}\)}.
\newblock In Ulrik Brandes and Dorothea Wagner, editors, {\em Graph-Theoretic
  Concepts in Computer Science, 26th International Workshop, {WG} 2000,
  Konstanz, Germany, June 15-17, 2000, Proceedings}, volume 1928 of {\em
  Lecture Notes in Computer Science}, pages 196--205. Springer, 2000.
\newblock \href {https://doi.org/10.1007/3-540-40064-8\_19}
  {\path{doi:10.1007/3-540-40064-8\_19}}.

\bibitem{Impagliazzo01}
Russell Impagliazzo and Ramamohan Paturi.
\newblock {On the Complexity of k-SAT}.
\newblock {\em J. Comput. Syst. Sci.}, 62(2):367--375, 2001.
\newblock \href {https://doi.org/10.1006/jcss.2000.1727}
  {\path{doi:10.1006/jcss.2000.1727}}.

\bibitem{sparsification}
Russell Impagliazzo, Ramamohan Paturi, and Francis Zane.
\newblock Which problems have strongly exponential complexity?
\newblock {\em J. Comput. Syst. Sci.}, 63(4):512--530, 2001.
\newblock \href {https://doi.org/10.1006/jcss.2001.1774}
  {\path{doi:10.1006/jcss.2001.1774}}.

\bibitem{Korhonen21}
Tuukka Korhonen.
\newblock Single-exponential time 2-approximation algorithm for treewidth.
\newblock {\em CoRR}, abs/2104.07463, 2021.
\newblock URL: \url{https://arxiv.org/abs/2104.07463}, \href
  {http://arxiv.org/abs/2104.07463} {\path{arXiv:2104.07463}}.

\bibitem{Oum08}
Sang{-}il Oum.
\newblock Approximating rank-width and clique-width quickly.
\newblock {\em {ACM} Trans. Algorithms}, 5(1):10:1--10:20, 2008.
\newblock \href {https://doi.org/10.1145/1435375.1435385}
  {\path{doi:10.1145/1435375.1435385}}.

\bibitem{Pilipczuk11}
Michal Pilipczuk.
\newblock Problems parameterized by treewidth tractable in single exponential
  time: {A} logical approach.
\newblock In Filip Murlak and Piotr Sankowski, editors, {\em Mathematical
  Foundations of Computer Science 2011 - 36th International Symposium, {MFCS}
  2011, Warsaw, Poland, August 22-26, 2011. Proceedings}, volume 6907 of {\em
  Lecture Notes in Computer Science}, pages 520--531. Springer, 2011.
\newblock \href {https://doi.org/10.1007/978-3-642-22993-0\_47}
  {\path{doi:10.1007/978-3-642-22993-0\_47}}.

\bibitem{Poljak74}
Svatopluk Poljak.
\newblock A note on stable sets and colorings of graphs.
\newblock {\em Commentationes Mathematicae Universitatis Carolinae},
  15(2):307--309, 1974.

\bibitem{Schirrmacher23}
Nicole Schirrmacher, Sebastian Siebertz, Giannos Stamoulis, Dimitrios~M.
  Thilikos, and Alexandre Vigny.
\newblock Model checking disjoint-paths logic on topological-minor-free graph
  classes.
\newblock {\em CoRR}, abs/2302.07033, 2023.
\newblock \href {http://arxiv.org/abs/2302.07033} {\path{arXiv:2302.07033}},
  \href {https://doi.org/10.48550/arXiv.2302.07033}
  {\path{doi:10.48550/arXiv.2302.07033}}.

\bibitem{Sintiari21}
Ni~Luh~Dewi Sintiari and Nicolas Trotignon.
\newblock (theta, triangle)-free and (even hole, k\({}_{\mbox{4}}\))-free
  graphs - part 1: Layered wheels.
\newblock {\em J. Graph Theory}, 97(4):475--509, 2021.
\newblock \href {https://doi.org/10.1002/jgt.22666}
  {\path{doi:10.1002/jgt.22666}}.

\bibitem{Torunczyk23}
Szymon Toruńczyk.
\newblock Flip-width: Cops and robber on dense graphs.
\newblock {\em CoRR}, abs/2302.00352, 2023.
\newblock \href {http://arxiv.org/abs/2302.00352} {\path{arXiv:2302.00352}},
  \href {https://doi.org/10.48550/arXiv.2302.00352}
  {\path{doi:10.48550/arXiv.2302.00352}}.

\bibitem{Cerny-07}
Jakub Černý.
\newblock Coloring circle graphs.
\newblock {\em Electronic Notes in Discrete Mathematics}, 29:457--461, 2007.
\newblock European Conference on Combinatorics, Graph Theory and Applications.
\newblock URL:
  \url{https://www.sciencedirect.com/science/article/pii/S1571065307001539},
  \href {https://doi.org/https://doi.org/10.1016/j.endm.2007.07.072}
  {\path{doi:https://doi.org/10.1016/j.endm.2007.07.072}}.

\end{thebibliography}

\end{document}